\documentclass[conference,10pt,final]{IEEEtran}
\IEEEoverridecommandlockouts

\usepackage{cite}
\usepackage{amsmath,amssymb,amsfonts}
\usepackage{algorithmic}
\usepackage{graphicx}
\usepackage{textcomp}
\usepackage{xcolor}
\def\BibTeX{{\rm B\kern-.05em{\sc i\kern-.025em b}\kern-.08em
    T\kern-.1667em\lower.7ex\hbox{E}\kern-.125emX}}

\usepackage{booktabs}  
\usepackage{subcaption} 
\usepackage[utf8]{inputenc}
\usepackage{amsthm}
\usepackage[ruled,vlined,linesnumbered]{algorithm2e}
\usepackage{stmaryrd}
\usepackage{amsmath}
\usepackage{enumerate}
\usepackage{hyperref}
\usepackage{thm-restate}
\usepackage{verbatim}
\theoremstyle{plain}
\usepackage[obeyFinal]{todonotes}
\usepackage{tikz}
\usetikzlibrary{automata, positioning, arrows,decorations.pathmorphing}
\usepackage{graphicx} 
\usepackage{xspace}
\usepackage{comment}
\usepackage{mathtools}

\usetikzlibrary{arrows.meta}

\newcommand{\pos}{\text{Pos}}

\newcommand{\dom}{\textsf{dom}}
\newcommand{\rr}{\textsf{RR}}

\newcommand{\RAT}{\ensuremath{\textsf{Rat}}\xspace}
\newcommand{\LEX}{\ensuremath{\textsf{Lex}}\xspace}
\newcommand{\SST}{\ensuremath{\textsf{SST}}\xspace}

\newcommand{\pref}{\ensuremath{\textsf{pref}}}
\newcommand{\rev}{\ensuremath{\textsf{rev}}\xspace}

\newcommand{\sub}{\ensuremath{\textsf{sub}}}
\newcommand{\era}{\ensuremath{\textsf{del}}}
\newcommand{\id}{\ensuremath{\textsf{id}}}
\newcommand{\MSOSI}{\textsf{MSOSI}\xspace}
\newcommand{\MSO}{\textsf{MSO}\xspace}
\newcommand{\MSOI}{\textsf{MSOI}\xspace}
\newcommand{\FOI}{\textsf{FOI}\xspace}
\newcommand{\FO}{\textsf{FO}\xspace}
\newcommand{\FOC}{\textsf{FOC}\xspace}
\newcommand{\MSOT}{\textsf{MSOT}\xspace}

\newcommand{\NMT}{\textsf{NMT}\xspace}
\newcommand{\SEQ}{\textsf{Seq}\xspace}
\newcommand{\REG}{\textsf{Reg}\xspace}
\newcommand{\PREG}{\textsf{PolyReg}\xspace}

\newcommand{\NMTSPF}{\ensuremath{\textsf{NMT}_{\textsf{spf}}}\xspace}

\newcommand{\deltac}{\delta_{\textsf{call}}}
\newcommand{\deltar}{\delta_{\textsf{ret}}}
\newcommand{\AT}{\textsf{AT}\xspace}
\newcommand{\suc}{\textsf{succ}}
\newcommand{\lexenum}{\textsf{lex{-}enum}}
\newcommand{\map}{\ensuremath{\textsf{map}}\xspace}
\newcommand{\maplex}{\textsf{maplex}}
\renewcommand{\square}{\ensuremath{\textsf{square}}\xspace}
\newcommand{\sep}{|}
\newcommand{\ch}{{\triangleright}}
\newcommand{\arity}{\mathsf{arity}}
\newcommand{\nat}{\mathbb N}

\newcommand{\partialf}{\rightharpoonup}
\newcommand{\univ}{\mathrm{univ}}
\newcommand{\set}[1]{\left\{#1\right\}}
\newcommand{\sign}{\mathcal S}
\newcommand{\signout}{\mathcal T}
\newcommand{\pad}{\text{\textvisiblespace}}
\newcommand{\order}{\prec}
\newcommand{\neworder}{\vartriangleleft}
\newcommand{\ordera}{\lambda}
\newcommand{\out}{\mu}
\newcommand{\aut}{A}
\newcommand{\ie}{\textit{i.e.~}}
\newcommand{\eg}{\textit{e.g.~}}

\newcommand{\ptime}{\textsc{PTime}\xspace}

\theoremstyle{definition}
\newtheorem{definition}{Definition}[section]
\newtheorem{example}[definition]{Example}

\theoremstyle{plain}
\newtheorem{theorem}{Theorem}
\newtheorem{conjecture}{Conjecture}
\newtheorem{question}{Question}
\newtheorem{lemma}[definition]{Lemma}
\newtheorem{proposition}[definition]{Proposition}
\newtheorem{corollary}[definition]{Corollary}
\newtheorem{remark}[definition]{Remark}

\theoremstyle{remark}

\newcommand{\mm}[2]{\begin{smallmatrix}#1\\#2\end{smallmatrix}}
\newcommand{\pp}[2]{\begin{psmallmatrix}#1\\#2\end{psmallmatrix}}

\newcommand{\mmm}[3]{\begin{smallmatrix}#1\\#2\\#3\end{smallmatrix}}

\begin{document}

\title{Lexicographic transductions of finite words}

\author{}

\author{\IEEEauthorblockN{Emmanuel Filiot}
\IEEEauthorblockA{\textit{Universit\'e libre de Bruxelles}}
\and
\IEEEauthorblockN{Nathan Lhote}
\IEEEauthorblockA{\textit{Aix-Marseille Universit\'e} }
\and
\IEEEauthorblockN{Pierre-Alain Reynier}
\IEEEauthorblockA{\textit{Aix-Marseille Universit\'e}}
}

\maketitle

\begin{abstract}
Regular transductions over finite words have linear input-to-output growth. This class of transductions enjoys many characterizations, such as transductions computable by two-way transducers 
    as well as transductions definable in \MSO (in the sense of Courcelle). Recently, regular transductions have been extended by Bojanczyk to polyregular transductions, which have polynomial growth, and are characterized by pebble transducers and \MSO interpretations. Another class of interest is that of transductions defined by streaming string transducers or marble transducers, which have exponential growth and are incomparable with polyregular transductions.
    
    In this paper, we consider \MSO set interpretations (\MSOSI) over finite words which were introduced by Colcombet and Loeding. \MSOSI are a natural candidate for the class of ``regular transductions with exponential growth", and are rather well-behaved.
    However \MSOSI lack, for now, two desirable properties that regular and polyregular transductions have. The first property is being described by an automaton model, which is closely related to the second property of regularity preserving meaning preserving regular languages under inverse image.
    
    We first show that if \MSOSI are (effectively) regularity preserving then any automatic $\omega$-word has a decidable \MSO theory, an almost 20 years old conjecture of B\'ar\'any.
    
    Our main contribution is the introduction of a class of transductions of exponential growth, which we call lexicographic transductions. We provide three different presentations for this class: first, as the closure of simple transductions (recognizable transductions) under a single operator called maplex;
    second, as a syntactic fragment of \MSOSI (but the regular languages are given by automata instead of formulas); and third, we give
    an automaton based model called nested marble transducers, which generalize both marble transducers and pebble transducers.
    We show that this class enjoys many nice properties including being regularity preserving.

\end{abstract}

\begin{IEEEkeywords}
transducers, automata, MSO, logical interpretations, automatic structures
\end{IEEEkeywords}

\section{Introduction}

\subsubsection{\MSOSI and the connection to automatic structures}

\MSO set interpretations (\MSOSI) were introduced in \cite{DBLP:journals/lmcs/ColcombetL07}, as a generalization of automatic structures (as well as $\omega$-automatic, tree automatic and $\omega$-tree automatic structures). Indeed an automatic structure can be seen as an \MSOSI whose domain is a single structure with decidable $\MSO$ theory such as $(\nat,\leq)$. Using a framework of transformations turns out to be very fruitful, and most of the properties of automatic structures already hold for set-interpretations over structures with decidable $\MSO$ theory. The core property of automatic structures (and their generalizations) is that they have decidable \FO theory. More generally, \MSOSI have what we call the \emph{\FO backward translation property}, meaning that the inverse image of an \FO formula by an \MSOSI is \MSO definable. This property is obtained via simple, yet powerful, syntactic formula substitution. This technique actually allows to show more generally that \MSOSI are closed under post-composition by \FO interpretations (\FOI).

Generally speaking, automatic structures do not have a decidable \MSO theory. This has motivated a line of research looking for interesting structures with a decidable \MSO theory. For instance morphic $\omega$-words, as well as two generalizations called $k$-lexicographic $\omega$-words \cite{DBLP:journals/ita/Barany08} and toric $\omega$-words \cite{DBLP:journals/tcs/BertheKNOVW25}, have been shown to have a decidable \MSO theory.
Morphic $\omega$-words and $k$-lexicographic $\omega$-words are particular cases of \emph{automatic $\omega$-words}\footnote{Not to be confused with automatic \emph{sequences}.}. An automatic $\omega$-word is an automatic structure with unary relations and a single binary relation which is a total order isomorphic to $(\nat,\leq)$ (it is crucial that the structure be given by its order relation and not by the successor).
To the best of our knowledge, it is not known whether an automatic $\omega$-word with an undecidable \MSO theory exists, which raises the following conjecture:
\begin{conjecture}\label{conj:automso}\cite[Conjecture~(1)~Section~9]{DBLP:journals/ita/Barany08}    
    Any automatic $\omega$-word has a decidable \MSO theory\footnote{B\'ar\'any actually conjectures more strongly that any automatic $\omega$-word has a so-called \emph{canonical presentation}.}. 
\end{conjecture}

In \cite[Corollary~5.6]{DBLP:journals/ita/Barany08}, the author even shows that $k$-lexicographic $\omega$-words are closed under sequential transductions. As we show in Proposition~\ref{prop:back-seq} this property is deeply connected to preserving \MSO definable sets by inverse image (which we call regularity preserving\footnote{This property is sometimes called regular continuity \cite{DBLP:journals/lmcs/CadilhacCP19}.}) and is stronger than having a decidable \MSO theory.

A different setting where one can obtain regularity preserving transductions, is 
provided in \cite{BKL19} where it is shown that \MSO interpretations (\MSOI) from finite words to finite words characterize the polyregular transductions. Once again, as for automatic $\omega$-words, the output structure must be defined by its order and not by the successor.

This calls for a more unifying argument and systematic study of \MSOSI whose output structures are linearly ordered, that we phrase as a conjecture\footnote{One could even venture stating stronger conjectures extending the structures to trees, $\omega$-words or infinite trees.}:
\begin{conjecture}\label{conj:mso}
    \MSOSI from finite words to finite words are regularity preserving.
\end{conjecture}
In this article we focus on transductions from finite words to finite words for two main reasons: it is already quite challenging and it captures part of the difficulty of $\omega$-words. Indeed we show in Corollary~\ref{coro:conje} that a positive answer to Conjecture~\ref{conj:mso} entails that Conjecture~\ref{conj:automso} holds.

\subsubsection{On regular transductions with exponential growth}

The theory of finite word transducers has a long history (in fact as long as automata theory) and is still actively studied.

Various classes of transductions have been introduced, most notably (and ordered inclusion-wise): sequential (\SEQ), rational (\RAT), regular (\REG) and the more recent polyregular transductions (\PREG) as well as transductions defined by streaming string transducers (\SST), which subsume \REG but are incomparable to \PREG. For a recent survey, see~\cite{DBLP:conf/stacs/MuschollP19}.

These classes are rather well-known and enjoy nice \emph{regularity} properties, including being closed under composition  (except for \SST which is still closed under post-composition by \SEQ) which entails\footnote{\label{note1}In Proposition~\ref{prop:back-seq} we see that the two are closely related.} being regularity preserving. However some important questions remain open, such as equivalence of \PREG transductions which is not known to be decidable.

The two classes of \REG and \PREG enjoy natural logical characterizations, namely word-to-word \MSO transductions (\MSOT) and \MSOI, respectively. The fact that \MSOT are regularity preserving is again obtained by simple formula substitution and holds for arbitrary structures. In contrast in the case of \MSOI, the only known proof is \textit{via} a translation into an automaton model called pebble transducers.
This raises a natural question for \MSOSI:

\begin{question}\label{ques:1}
    Can one obtain an automaton model corresponding to \MSOSI over finite words ?
\end{question}
A positive answer to this question would hopefully provide a proof of Conjecture~\ref{conj:mso}, since natural automata models are usually closed under post-composition by \SEQ\textsuperscript{\ref{note1}}.
While hope plays an important part in research, we have good reasons to think this is a hard problem: as mentioned above this would solve a long standing open problem on automatic structures.

It is rather clear that \PREG captures the ``right" notion of regular transductions with polynomial growth. While \MSOSI seems like a natural candidate, not enough is known about this class yet to say that it captures the ``right" notion of regular transductions with exponential growth.

Let us more humbly describe what should be, in our view, a \emph{nice} class of regular transductions with exponential growth:
this class should be characterized by different, somewhat natural\footnote{As opposed to an artificial model like the union of \PREG with \SST.}, computation models which subsume the well-behaved classes of \PREG and \SST. It should be regularity preserving and potentially\textsuperscript{\ref{note1}} have extra closure properties by pre- or post-composition with smaller classes. In this article we introduce the class of \emph{lexicographic transductions}  (\LEX) which meets all the above criteria.

\subsubsection{Contributions}

The first contribution of the article is a hardness result: showing that word-to-word \MSOSI are regularity preserving is at least as hard as showing that any automatic $\omega$-word has a decidable \MSO theory (Corollary~\ref{coro:conje}). To obtain this result we define \emph{automatic transduction} (\AT) which are naturally equivalent to \MSOSI but formulated in a way that makes the connection with automatic structures clearer. That way we obtain a one-to-one correspondence between automatic $\omega$-words and automatic transductions over a unary alphabet which define a total function (Proposition~\ref{prop:corr}).

The main contribution of the article is the introduction of a new class of transductions, called \emph{lexicographic transductions} (\LEX). We give three different characterizations of this class and show that it enjoys many nice properties, including being regularity preserving.

The first definition of \LEX is in the spirit of list functions of \cite{DBLP:conf/lics/BojanczykDK18,DBLP:polyreg}: we start with \emph{simple functions} which are recognizable transductions whose range contains words of length at most $1$ only. Then we close the class under a single type of operator called \maplex\ which works as follows: $\maplex\ f$ maps a word $u$ to the concatenation $f(u_1)f(u_2)\dots f(u_n)$ where $u_1,\dots,u_n$ are all the labellings of $u$ over some fixed and totally ordered alphabet, enumerated in lexicographic order.

Secondly, we show that this class can be expressed as a syntactic restriction of \AT, which we call lexicographic automatic transductions ($\AT_\LEX$). These two characterizations are actually syntactically equivalent but quite different in spirit.
We leverage the aforementioned correspondence between automatic $\omega$-words and automatic transduction, as well  as a result of B\'ar\'any to show that the nesting of \maplex\ operators generates a strict hierarchy of transduction classes (Corollary~\ref{coro:stricthierarch}).

Thirdly, we introduce an automaton model called \emph{nested marble transducers} (\NMT).
Nested marble transducers are quite expressive: they generalize marble transducers~\cite{DBLP:journals/actaC/EngelfrietHB99,DBLP:conf/mfcs/Doueneau-TabotF20} which are known to coincide with \SST, they also naturally generalize \PREG. Informally, a level $k$ nested marble transducer can annotate its input as a marble transducer (i.e. it drops a marble whenever moving left and lifts a marble whenever moving right), and call a level $k{-}1$ nested marble transducer to run on this annotated configuration. This call returns both an output string and a state which the top-level transducer can use to take its next transition. This passing of information from the lower levels to the higher levels is what allows
to prove strong closure properties of \NMT: we show that \NMT have regular domains, are regularity preserving, and more generally are closed under post-composition by \PREG. 

Regarding expressiveness, \LEX can be expressed by \NMT in a rather direct way (Theorem~\ref{thm:lex2marble}).
In the other way, transductions expressed in \LEX do not have such a state-passing mechanism, hence showing that 
\NMT is included in \LEX constitutes the technical heart of this article (Theorem~\ref{thm:marble2lex}). An important step consists in
showing that one can remove the state-passing mechanism in \NMT (Theorem~\ref{thm:sp-removal}). On top of being technical, we show it is computationally costly: there is an unavoidable non-elementary blow-up to transform a nested marble transducers into nested marble transducers without state-passing.

\subsubsection{Outline of the paper}
Preliminaries on languages and transductions
are given in Section~\ref{sec:prel}, we then
detail the definitions of structures and interpretations,
as well as their connection to automatic structures,
in Section~\ref{sec:msosi}. We present the class of
lexicographic transductions in Section~\ref{sec:lex}
and detail examples in Section~\ref{sec:examples}.
The model of nested marble transducers is presented
in Section~\ref{sec:nmt}, its equivalence with \LEX is proven
in Section~\ref{sec:equivalence}, and its closure properties
are detailed in
Section~\ref{sec:ppties}. Proof details
can be found in the Appendix.

\section{Word languages and transductions}
\label{sec:prel}

We let $\mathcal{P}(X)$ the powerset of any set $X$. 

\paragraph{Words and languages} Given an alphabet $\Sigma$, a
$\Sigma$-word $u$ (or just word if $\Sigma$ is clear from the context) is a sequence of letters from
$\Sigma$. We denote by $\epsilon$ the empty word, and by $|u|$ the
length of a word $u$. In particular $|\epsilon|=0$. For all integers
$n\geq 0$, we let $\Sigma^n$ (resp. $\Sigma^{\leq n}$) be the set of
words of length $n$ (resp. at most $n$). We let $\pos(u) =
\{1,\dots,|u|\}$ be the set of \emph{positions} of $u$, and for all
$i\in\pos(u)$, $u[i]\in\Sigma$ is the $i$-th letter of $u$. We write
$\Sigma^*$ for the set of words over $\Sigma$, and $\Sigma^+$ for the
set of non-empty words. A word language over $\Sigma$ is a subset of
$\Sigma^*$. \emph{In this paper, we let $\sep$ be a symbol called
  separator, assumed to be distinct from any alphabet symbol.}

\paragraph{Convolution} Let $\Sigma_1,\Sigma_2$ be two alphabets,  $\ell\in\nat$
and $u_1\in\Sigma_1^\ell,u_2\in\Sigma_2^\ell$ be two words of length $\ell$. The \emph{convolution} of $u_1$ and $u_2$ is the word
of length $\ell$ over $\Sigma_1\times \Sigma_2$ denoted $u_1\otimes
u_2$, such that for all $1\leq i\leq \ell$, $(u_1\otimes u_2)[i] =
(u_1[i],u_2[i])$.

\paragraph{Finite automata} A (non-deterministic) finite automaton (NFA) over an alphabet $\Sigma$ is denoted as a tuple $A=(Q, q_0, F, \Delta)$ where $Q$ is the set of states, $q_0$ the initial state, $F\subseteq Q$ the final states, and $\Delta\subseteq Q\times \Sigma\times Q$ the transition relation. We write $q\xrightarrow{u}_A q'$ when there exists a run of $A$ from state $q$ to state $q'$ on $u$, and denote by $L(A) = \{ u{\in}\Sigma^*\mid q_0\xrightarrow{u}_A q_f\in F\}$ the language recognized by $A$. When $A$ is a deterministic finite automaton (DFA), the transition relation is denoted by a (partial) function $\delta : Q\times \Sigma \partialf Q$.

\paragraph{Word transductions} A \emph{word transduction} (or just transduction for short) over $\Sigma,\Gamma$ two alphabets is a (partial) function $f : \Sigma^*\partialf \Gamma^*$. We denote by $\dom(f)$ its domain. Given two transductions
$f_1,f_2:\Sigma^*\partialf \Gamma^*$ with disjoint domains, we
let $f_1+f_2$ be the transduction of domain $\dom(f_1)\cup \dom(f_2)$
such that $(f_1+f_2)(u) = f_i(u)$ if $u\in \dom(f_i)$. 
Given $f : \Sigma^*\partialf \Gamma^*$, $g:\Gamma^*\partialf \Lambda^*$, we write $(g\ f) : \Sigma^*\partialf \Lambda^*$ the composition $g\circ f$. Given  $h:\Lambda^*\partialf \Delta^*$, $(h\ g\ f)$ stands for $(h\ (g\ f))$. For $u\in\Sigma^*$, we also write $(h\ g\ f\ u)$ for $(h\ g\ f)(u)$.

A transduction $f$ has \emph{exponential growth} if there exists $c\in\mathbb{N}$ such that for all $u\in\dom(f)$, $|f(u)|\leq 2^{c|u|}$ holds.
A transduction $f$ has \emph{polynomial growth} if there exists $c,k\in\mathbb{N}$ such that for all $u\in\dom(f)$, $|f(u)|\leq c|u|^k$ holds.

\begin{example}[Reverse and copy]
Let $\Sigma$ be an alphabet. The transduction $\rev:\Sigma^*\rightarrow\Sigma^*$ takes as input any word  
$u=\sigma_1\dots \sigma_n$ and outputs its reverse $\sigma_n\dots \sigma_1$,  for all
    $\sigma_i\in\Sigma$. The transduction $\textsf{copy}$ takes $u$ and returns $uu$. 
\end{example}

\begin{example}[Square]\label{ex:squareex}
    Let $\Sigma$ be an alphabet and $\underline{\Sigma} =
    \{\underline{\sigma}\mid \sigma\in\Sigma\}$. Given a word $u=\sigma_1\dots\sigma_n$ and
    a position $i\in\pos(u)$, we let $\textsf{under}_i(u) =\sigma_1\dots
    \sigma_{i{-}1}\underline{\sigma_i}\sigma_{i+1}\dots \sigma_n$.
The transduction
    $\square:\Sigma^*\rightarrow (\Sigma\cup\underline{\Sigma})^*$ is
    defined as
    $\square(u) = \textsf{under}_1(u)\dots \textsf{under}_{|u|}(u)$.
    For example $\square(abc) =
    \underline{a}bca\underline{b}cab\underline{c}$.
\end{example}

\begin{example}[Map]\label{ex:map}
Let $\Sigma$ be some alphabet and $|\not\in \Sigma$ be some
    separator symbol. Let $\Sigma_\sep = \Sigma\cup \{\sep\}$. 
Let $f : \Sigma^*\partialf\Gamma^*$. The transduction $\map\ f:\Sigma_|^*\partialf \Sigma_|^*$ takes any input word of the
    form $u = u_1\sep u_2\sep\dots \sep u_n$ where $u_i\in\Sigma^*$ for all
    $i\in\{1,\dots,n\}$, and returns $f(u_1)\sep f(u_2)\sep\dots
    \sep f(u_n)$ (if all the $f(u_i)$ are defined, otherwise $(\map\ f)(u)$ is undefined.
\end{example}

\paragraph{Sequential and rational transductions} Sequential transductions are transductions recognized by sequential transducers. A \emph{sequential transducer} over some alphabets $\Sigma$ and $\Gamma$ (not necessarily disjoint), is a pair $T = (A, \out)$ where $A = (Q,q_0,F,\delta)$ is a DFA over $\Sigma$ and $\out : \dom(\delta)\rightarrow \Gamma^*$ is a total function. We write $q\xrightarrow{u/v}_T q'$ whenever there exists a sequence of states $q_1=q, q_2,\dots, q_{n+1} = q'$ such that
$q_1\xrightarrow{u[1]}_A q_2\dots q_n\xrightarrow{u[n]}_A q_{n+1}$
where $n=|u|$, and $v = \out(q_1,u[1])\dots \out(q_n,u[n])$. The
transduction $f_T$ recognized by $T$ is defined for all $u\in L(A)$ by
$f_T(u) = v$ such that $q_0\xrightarrow{u/v} q_f\in F$. Note that
$\dom(f_T) = L(A)$. We denote by $\SEQ$ the class of sequential
transductions. Like sequential transducers, a (non-deterministic, functional)
\emph{finite state transducer}\footnote{This class is also called
  real-time finite state transducers in the literature.}
  is defined as a pair $T = (A,\out)$ but
$A$ can be non-deterministic, with the \emph{functional} restriction: for all words
$u\in L(A)$, the
outputs of all the accepting runs over $u$ are all equal. With this
restriction, $T$ recognizes a transduction $f_T$. A \emph{rational}
transduction is a transduction $f_T$ for some $T$, and we denote by
$\RAT$ the class of rational transductions~\cite{DBLP:books/ems/21/HarjuK21}.

\paragraph{Regular and polyregular transductions} The class of
regular (resp. polyregular) transductions is defined as the smallest class of
transductions which is closed under composition of transductions and
\map, and contains the sequential transductions, \textsf{copy} and \rev (resp. the sequential
transductions, \rev and
\square)~\cite{DBLP:conf/icalp/BojanczykS20,DBLP:polyreg}. We denote by \textsf{PolyReg} the class of polyregular transductions.

\section{\MSO set interpretations, properties and limitations}
\label{sec:msosi}
\subsubsection{\MSO set interpretations}

\paragraph{Signatures, formulas and structures} 
A \emph{relational signature} (or simply signature) is a set $\sign$ of symbols together with a function $\arity: \sign\rightarrow \nat$.

We consider a set of \emph{first-order variables} denoted by lower case letter $x,y,z,\ldots$ as well as a set of \emph{second-order variables} denoted by upper case letter $X,Y,Z\ldots$.
The \MSO-formulas over signature $\sign$, denoted by $\MSO [\sign] $
are given by the following grammar $\phi::=$
$$\exists x \phi \mid \exists X \phi \mid  \phi\wedge \phi \mid \neg \phi \mid X(x) \mid R(x_1,\ldots,x_r)$$


where $x,x_1,\ldots, x_r$ are first-order variables, $X$ is a second-order variable and $R\in \sign$ with $\arity(R)=r$.
We denote by $\FO[\sign]$ the formulas which don't use second-order variables.

A \emph{relational structure} $u$ over signature $\sign$ is a set $U$ called the \emph{universe} of the structure, together with, for each symbol $R\in\sign$ of arity $r$, an \emph{interpretation} denoted $R^u$ which is a subset of  $U^r$.

\paragraph{Regularity preserving} A function from $\sign$-structures to $\signout$-structures is called \emph{regularity preserving} if the inverse image of an $\MSO[\signout]$ definable set is $\MSO[\sign]$ definable. We say that a class of functions is regularity preserving if all functions in the class are.

\paragraph{Word structures} The \emph{word signature over $\Sigma$} is the tuple $S_\Sigma =
((\sigma(x))_{\sigma\in \Sigma}, \leq(x,y))$ where $\sigma(x)$ are unary
predicate symbols and $\leq(x,y)$, usually written $x\leq y$, is a
binary predicate symbol. Any word $u$ can be naturally associated with
an $S_\Sigma$-structure $\tilde{u} = (U,(\sigma^{\tilde{u}})_{\sigma\in\Sigma},\leq^{\tilde{u}})$ where
$U = \pos(u)$, $\sigma^{\tilde{u}}$ is a set of positions labeled $\sigma$, for
all $\sigma\in\Sigma$, and $\leq^{\tilde{u}}$ is the natural (linear) order on
$\pos(u)$.  We write $u$ instead of $\tilde{u}$ if it is clear from the
context that $u$ is an $S_\Sigma$-structure. A \emph{word structure} over $\Sigma$ is an
$S_\Sigma$-structure isomorphic to some $\tilde{u}$. 
Note that being a word structure is \FO definable.

\paragraph{\MSO set interpretations}
We define \MSO set interpretations as in~\cite{DBLP:journals/lmcs/ColcombetL07}.

\begin{definition}
    An \MSO set interpretation (\MSOSI) $T$ from $\sign$-structures to $\signout$-structures, is given by

    \begin{itemize}
        \item $k\in \nat\setminus \set{0}$ called the dimension,
        \item a domain formula $\phi_\dom\in \MSO[\sign]$,
        \item an output universe formula $\phi_\univ(\overline X)\in\MSO[\sign] $,
        \item for each relation symbol $R\in \signout$ of arity $r$ a formula $\phi_R(\overline{X_1},\ldots,\overline{X_r})\in \MSO[\sign]$
    \end{itemize}

    where $\overline{X},\overline{X_1},\ldots$ are $k$-tuples of variables.

    We now define the semantics of $T$ which is a partial transduction $f_T$ from $\sign$-structures to $\signout$-structures. The \emph{domain} of $f_T$ is the set of structures $u$ such that $u\models \phi_\dom$. Given such a $u$ with universe $U$, we define its image $v=f_T(u)$:

    \begin{itemize}
        \item The universe of $v$ is the set $V = \{ \overline{P}\in \mathcal P(U)^k\mid u\models
          \phi_{univ}(\overline{P})\}$,
          \item for $R\in \signout$ of arity $r$, $R^{v} = \{ (\overline{P_1},\ldots,\overline{P_r})\in V^k\mid u\models
            \phi_{R}(\overline{P_1},\ldots,\overline{P_r})\}$.
    \end{itemize}

    We say that an \MSOSI is (finite) word-to-word if its domain and co-domain only contain word structures over some respective alphabets $\Sigma,\Gamma$.
\end{definition}

\begin{remark}
Given an \MSOSI from $\sign_\Sigma$-structures to $\sign_\Gamma$-structures, one can restrict the domain formula to word structures whose image are word structures. This is because being a linear order is \FO-definable.
\end{remark}

\paragraph{\MSO transductions, \MSO and \FO interpretations}
An \MSO interpretation (\MSOI) is an \MSOSI whose free set variables are restricted to be singleton sets. This can be syntactically enforced in the universe formula $\phi_{univ}$, as being a singleton is an \MSO definable property. Equivalently, \MSOI are defined as \MSOSI but instead the free variables are first-order. Note that transductions realized by \MSOI have only polynomial growth. An \FO interpretation (\FOI) is an \MSOI whose formulas are all \FO-formulas.
Finally an \MSO transduction (\MSOT) is (roughly\footnote{Classically, one adds a bounded number of copies of the input to get the full class of \MSOT.} speaking) an \MSO interpretation of dimension $1$.
\todo{Mention that MSOT captures exactly regular transductions?}

The following theorem is at the core of the theory of set interpretations, and automatic structures. It holds in all generality, and furthermore the compositions can be done by simple formula substitutions.
\begin{theorem}\cite[Proposition~2.4]{DBLP:journals/lmcs/ColcombetL07}
    \MSO set interpretations are effectively closed under pre-composition by \MSOT and post-composition by \FOI.
\end{theorem}

\subsubsection{Properties of word-to-word set interpretations}

\paragraph{Exponential versus polynomial growth}
There is a dichotomy for the growth of set interpretations defined over words, deeply connected to the similar dichotomy for the ambiguity of automata, between \emph{exponential growth} and \emph{polynomial growth}. Moreover for polynomial growth transductions, the level of growth exactly coincides with the minimum dimension of an \MSOSI defining the transduction.
The result can be obtained in the more general case of trees.
\begin{theorem}\cite[Theorem~1.5]{GLN2025},\cite[Theorem~2.3]{Bojanczyk23a}\label{thm:growth}
    A set interpretation over words has growth either $2^{\Theta(n)}$, or $\Theta(n^k)$ for some $k\in \nat$, and this can be computed in \ptime.
    In the latter case\footnote{Note that to get this tight correspondence, we need to allow a bounded number of copies of the input, see \cite[Definition~4.3]{GLN2025}.}, one can compute an equivalent \MSOI of dimension $k$.
\end{theorem}

Quite a lot is known about word-to-word set interpretations with polynomial growth, which are called polyregular transductions and enjoy many different characterizations~\cite[Theorem~7]{BKL19}.

\paragraph{Regularity preserving}

An open question on word-to-word \MSOSI is whether they are regularity preserving. This can actually be formulated in terms of closure properties.
\begin{restatable}{proposition}{backseq}
\label{prop:back-seq}
    The following are equivalent:
    \begin{itemize}
        \item Word-to-word \MSOSI are regularity preserving,
        \item The class of word-to-word \MSOSI is closed under post-composition with transductions computed by Mealy machines,
        \item The class of word-to-word \MSOSI is closed under post-composition with polyregular transductions.
    \end{itemize}
\end{restatable}

\subsubsection{Automatic transductions}
We describe an automata-based presentation of \MSOSI, which we call \emph{automatic transductions}. This presentation has two main advantages: firstly it is more amenable to efficient processing, as it is based on automata instead of \MSO. Secondly it makes the connection between automatic structures and set interpretations more obvious.

 \begin{definition}
 An automatic transduction (\AT for short) from $\Sigma^*$ to $\signout$-structures is described as a tuple $T=(\Sigma,B,\aut_\dom,\aut_\univ, (\aut_R)_{R\in \sign})$
 where:
 \begin{itemize}
 \item $B$ is a finite alphabet describing a work alphabet
 \item $\aut_\dom$ is an automaton over $\Sigma$ recognizing the domain of the transduction,
 \item $ A_\univ$ is an automaton over $\Sigma\times B$ describing the possible configurations,
 \item for each $R\in \sign$ of arity $r$, $\aut_R $ is an automaton over $\Sigma\times B^r$ describing tuples of the relation $R$.
 \end{itemize}

 \noindent\emph{Semantics} Given a word $u\in \Sigma^*$, the output $\signout$-structure $v=f_T(u)$ is defined, whenever $u\in L(\aut_\dom)$, as follows:
 \begin{itemize}
     \item its universe is the set $V=\set{x\in B^*|\ u\otimes x\in L(\aut_\univ)}$,
     \item a predicate symbol $R\in \signout$ of arity $r$ is interpreted as:
     
     $R^v=\set{(x_1,\ldots, x_r)\in V^r|u\otimes x_1\otimes \ldots \otimes x_r\in L(\aut_R)}$.
 \end{itemize}

  \end{definition}

\begin{remark}~

\begin{itemize}
    \item Automatic transductions are essentially identical to \MSOSI, except restricted to input word structures, where one can leverage the classical equivalence between \MSO and automata.

    \item Automatic transductions can be naturally generalized to work over input structures such as $\omega$-words, trees and infinite trees, giving rise the notions of $\omega$-automatic, tree-automatic and $\omega$-tree-automatic transductions.
    Note that an $\omega$-automatic structure is precisely an $\omega$-automatic transduction whose domain is a single infinite word $a^\omega$.
    \end{itemize}
\end{remark}

\subsubsection{Transduction/structure correspondence}
\label{subsec:msosi}
Here we give a rather natural correspondence between word-to-word automatic transductions and automatic $\omega$-words.

We use definitions inspired from \cite{DBLP:journals/ita/Barany08}. We extend the convolution operation to words of different lengths by adding a padding symbol $\pad$ to the right of the shortest word until the lengths match.

\begin{definition}

An \emph{automatic structure} over the signature $\signout$ is given by a tuple $S=(B,\aut_\univ,(\aut_R)_{R\in  \sign})$, $\aut_\univ$ is an automaton over alphabet $B$ and for $R\in \sign$ of arity $r$, $\aut_{R}$ is an automaton over the alphabet $(B{\cup}\set{\pad})^r$. The structure $u$ associated with $S$ has universe $U=L(\aut)$, the interpretation of $R\in \sign$ is the set $R^u=\set{(x_1,\ldots,x_r)\in V^r|\ x_1\otimes \cdots \otimes x_r\in L_R}$.
An automatic structure which is an $\omega$-word structure over some alphabet is called an automatic $\omega$-word.
\end{definition}

\paragraph{Length monotonous} An automatic presentation of an $\omega$-word is \emph{length monotonous} if shorter words always appear before longer words with respect to the linear order. It turns out that any automatic $\omega$-word can be presented in a length monotonous way: this is done by adding an extra padding letter $\sharp$, and padding any word $v$ with enough symbols so that it is longer than any word appearing before in the order. The proof of this is rather simple and left to the reader. In the following, we shall assume that all automatic presentations of $\omega$-words are length monotonous.

Given a total transduction over a unary alphabet $f:a^*\rightarrow \Sigma^*$, we define its \emph{product} $\Pi f=f(\varepsilon)\cdot f(a)\cdot f(aa)\cdots$. Note that $\Pi f$ may be a finite word.

\begin{proposition}\label{prop:corr}
    There is a one-to-one correspondence between length monotonous automatic presentations of $\omega$-words and total automatic transductions from finite words over a unary alphabet to finite words. Moreover, when an automatic $\omega$-word $w$ corresponds to a function $f$, then $w=\Pi f$.
\end{proposition}
\begin{proof}
    Let $S=(\aut_\univ,(\aut_\sigma)_{\sigma\in \Sigma},\aut_{\leq})$ be a length monotonous automatic presentation of a word $w$.
    Since $S$ is length monotonous, we only care about the part of $\aut_{\leq}$ which reads letters over $B\times B$, and we define $\aut_{\leq}'$ the automaton obtained by removing transitions which use a $\pad$ padding symbol.

    We define the automatic transduction over $\set{a}$, $T=(\top,\aut_\univ,(\aut_\sigma)_{\sigma\in \Sigma},\aut_{\leq}')$ realizing a transduction $f$. By construction, $w=\Pi f$.

    Conversely, let $T=(\top,\aut_\univ,(\aut_\sigma)_{\sigma\in \Sigma},\aut_{\leq})$ realize a total transduction $f:\set{a}\rightarrow \Sigma^*$. If $\aut_{\leq}$ is over $B\times B$, we extend it to an automaton $\aut_{\leq}$ in a length monotonous way to convolutions of words of possibly different lengths.
    Again by construction the automatic $\omega$-word $w$ presented by $S=(\aut_\univ,(\aut_\sigma)_{\sigma\in \Sigma},\aut_{\leq}')$ satisfies $w=\Pi f$.

    The two constructions are indeed inverse of each other.
\end{proof}

We can now prove that if Conjecture~\ref{conj:mso} holds then Conjecture~\ref{conj:automso} also holds.

\begin{corollary}\label{coro:conje}
    If word-to-word \MSOSI are regularity preserving, then automatic $\omega$-words have a decidable \MSO theory\footnote{One could actually prove the stronger implication that any automatic $\omega$-word has a canonical presentation, as in \cite{DBLP:journals/ita/Barany08}.}.
\end{corollary}
\begin{proof}
    Assume that word-to-word \MSOSI are regularity preserving.
    Let $u\in \Sigma^\omega$ be an automatic $\omega$-word and let $L\subseteq \Sigma^\omega$ be a regular language.
    We define $f:\set{a}^*\rightarrow \Sigma^*$ the automatic transduction given by Proposition~\ref{prop:corr}.

    Let $\aut$ be a Büchi automaton recognizing $L$, let $Q$ be the set of states, $q_0$ the initial state and $F\subseteq Q$ be the set of accepting states.
    For $p,q\in Q$, we define $L_{p,q}^1$ the set of finite words which have a run from $p$ to $q$ visiting at least one final state and $L_{p,q}^0$ the set of finite words which have a run from $p$ to $q$ visiting no final state.
    We define a formula which starts with a block of existential quantification over variables $X_p^b$ with $p\in Q$ and $b\in \set{0,1}$. 
    The idea is that an integer $x$ will belong to $X_p^b$ if there is a run of $\aut $ over $\Pi_{i\leq x}f(a^i)$ which ends in state $p$ and sees at least $b$ final state \emph{after} having read  $\Pi_{i< x}f(a^i)$.
    The formula will be a conjunction of four parts: firstly all the $X_p^b$'s are pairwise disjoint. Secondly there is a $p$ and an infinite set of positions $x$ such that $X_p^1(x)$ holds.
    The third part is here to ensure that for the minimal $x$, if $X_p^b(x)$ then $f(a^1)\in L_{q_0,p}^b$ can reach state $p$ by visiting at least $b$ final states. It is a conjunction for all $p\in Q$, $b\in \set{0,1}$ of $\forall x\  (X_p^b(x)\wedge \forall y\ x\leq y) \rightarrow f(a)\in L_{q_0,p}^b$. Here $f(a)\in L_{q_0,p}^b$ is simply a boolean which can be computed.
    The last part of the formula will be the conjunction for all $p,q\in Q$, $b,c\in \set{0,1}$ of $\forall x\ X_p^b(x)\wedge X_q^c(x+1) \rightarrow f(a^x)\in L_{p,q}^c$.
    The fact that $f(a^x)\in L_{p,q}^c$ is definable in \MSO comes from the assumption that \MSOSI are regularity preserving.

    We have transferred an \MSO property of $u$ to an \MSO property of $(\nat,<)$, which is decidable.
\end{proof}

\section{Lexicographic transductions}\label{sec:lex}

As we have seen in the latter section, we do not know whether $\MSOSI$
are regularity preserving, and as a consequence of
Corollary~\ref{coro:conje}, proving that it enjoys this property
would prove a long-standing conjecture of the theory of automatic
structures. In this section, we introduce a subclass of \MSOSI transductions which
enjoys this property, called
\emph{lexicographic transductions}.

\subsubsection{Definition of lexicographic transductions}

We first define this class in
terms of closure of basic transductions, called \emph{simple
  transductions}, under a \emph{lexicographic map} operation. The
connection with \MSOSI is done at the end of this section (paragraph~\ref{subsec:lexaut}), via a
corresponding subclass of automatic transductions. 

\paragraph*{Simple transductions}
A  \emph{regular constant} (partial) transduction of type $\Sigma^*\partialf \Gamma^*$ can be
denoted by an expression of the form $L\ch w$, where $L$ is a
regular language over $\Sigma$ and $w$ is a word in $\Gamma^*$, such that for
all $u\in\Sigma^*$, $(L\ch w)(u)$ is defined only if $u\in L$, by
$(L\ch w)(u) = w$. 
A \emph{simple transduction}\footnote{It is a restriction of the known class of
  \emph{recognizable transduction} to output words of length at most $1$.} $f$ is a finite union of regular constant transductions whose codomain only contains words of length at most $1$.
A simple transduction $f:\Sigma^*\partialf \Gamma^*$ is denoted by $f=\sum_{i=1}^n L_i\ch w_i$ such that
 $L_1,\dots,L_n\subseteq \Sigma^*$ are pairwise disjoint regular languages, and $w_1,\dots,w_n\in\Gamma^{\leq 1}$.

\paragraph*{Lexicographic enumerators}

    An ordered alphabet is a pair $\ordera=(B,\order)$ such that $B$ is finite set and $\order$ is a linear order over $B$.
The order $\order$ is extended lexicographically (using the same notation) to words  of same length over $B$, with most significant letter to the right: for all $n$ and all $u,v\in B^n$, $u\order v$ if there exists a position $i\leq n$ such that
$u[i]\order v[i]$ and for all $i<j\leq n$, $u[j]=v[j]$. Note that $\order$ is a total order over $B^n$, for all $n$. We denote $\suc_\ordera:B^*\partialf B^*$ the successor function on $B^*$ induced by $\order$.

We recall that $|$ is a fixed
separator symbol. The \emph{$\ordera$-lexicographic enumerator} is the
function $\lexenum_{\ordera}$:
$$
\begin{array}{llllllll}
   \bigcup_{\Sigma,\Gamma\text{ alphabets}} & \Sigma^* & \rightarrow & ((\Sigma\times
                                                        B)^*|)^* \\
   & w & \mapsto & (w\otimes u_1)|(w\otimes
u_2)|\dots | (w\otimes u_k)
  \end{array}
$$
where $|w| = |u_1|=\dots=|u_k|$, $u_1$ is minimal for $\order$, $u_k$
is maximal for $\order$ and for all $1\leq i<k$, $u_{i+1} =
\suc_\ordera(u_i)$. Note that $k = |B|^{|w|}$.

\begin{example}
Let $\Sigma=\{a,b\}$ let $\ordera=(B,\order)$ be a finite order with $B=\{0,1\}$ and $0\order 1$. 
For all $\sigma\in\Sigma$ and $b\in B$, we write
    $\mm{\sigma}{b}$ instead of $(\sigma,b)$ and $\mm{\Sigma}{b}$
  to denote the set of pairs $(\sigma,b)$ for all $\sigma\in \Sigma$. Then
$\lexenum_{\ordera}(abb) =$
$$
\mm{a}{0}\mm{b}{0}\mm{b}{0}|
\mm{a}{1}\mm{b}{0}\mm{b}{0}|
\mm{a}{0}\mm{b}{1}\mm{b}{0}|
\mm{a}{1}\mm{b}{1}\mm{b}{0}|
\mm{a}{0}\mm{b}{0}\mm{b}{1}|
\mm{a}{1}\mm{b}{0}\mm{b}{1}|
\mm{a}{0}\mm{b}{1}\mm{b}{1}|
\mm{a}{1}\mm{b}{1}\mm{b}{1}
$$
\end{example}

\paragraph*{MapLex combinator} Let $\ordera=(B,\order)$ be an ordered alphabet. We define the function
$$
\begin{array}{llllll}
\maplex_{\ordera} & : & \bigcup_{\Sigma,\Gamma\text{ alphabets}}((\Sigma\times B)^*\rightarrow \Gamma^*)\rightarrow \Sigma^*\rightarrow \Gamma^*\\
\end{array}
$$
such that for all $\Sigma,\Gamma$ alphabets, all $f:(\Sigma\times B)^*\rightarrow \Gamma^*$  and
$u\in\Sigma^*$, 
$$
\maplex_{\ordera}\ f\ u\ =\ f(v_1)f(v_2)\dots f(v_k)
$$
where $\lexenum_{\ordera}(u) = v_1|v_2|\dots |v_k$.
Note that $u$ is in the domain of $\maplex_{\ordera}\ f$ if and only if $v_1,\ldots,v_k$ are all in the domain of $f$.
We write $\maplex$ when $\order$ is clear from the
context.  \todo{encore utile?}

\begin{definition}[Lexicographic transductions]
Lexicographic transductions, denoted by  $\LEX$, are defined inductively by $\LEX_0$ the class of simple transductions and $\LEX_{k+1}=\set{\maplex_{\ordera}\  f| f\in \LEX_k,\ \ordera\ \text{ordered alphabet}}$.

\end{definition}

\begin{example}[Identity and Reverse]
Take $\ordera=(B,\order)$ and $\ordera'=(B,\order')$ with $B=\{0,1\}$, $0\order 1$, and $1\order'
    0$.  
  For all $\sigma\in\Sigma$, let $L_\sigma =
  \pp{\Sigma}{0}^*\pp{\sigma}{1}\pp{\Sigma}{0}^*$ and $L_\epsilon=(\Sigma\times B)^*\setminus(\bigcup_\sigma L_\sigma)$.
    $$
    \begin{array}{rcl}
      \textsf{id} & = & \maplex_{\ordera}\ (L_\epsilon\ch \epsilon +
    \sum_{\sigma\in \Sigma}L_\sigma\ch \sigma) \\ 
    \rev & = & \maplex_{\ordera'}\ (L_\epsilon\ch \epsilon +
    \sum_{\sigma\in \Sigma}L_\sigma\ch \sigma)
      \end{array}
      $$

This is illustrated below on input $abc$, with the output of the simple function below every word of the enumeration.
\renewcommand{\arraystretch}{1.5}
      $$      
      \begin{array}{ll}
\id & \underbrace{\mm{a}{0}\mm{b}{0}\mm{c}{0}}_{\epsilon}|
\underbrace{\mm{a}{1}\mm{b}{0}\mm{c}{0}}_{a}|
\underbrace{\mm{a}{0}\mm{b}{1}\mm{c}{0}}_{b}|
\underbrace{\mm{a}{1}\mm{b}{1}\mm{c}{0}}_{\epsilon}|
\underbrace{\mm{a}{0}\mm{b}{0}\mm{c}{1}}_{c}|
\underbrace{\mm{a}{1}\mm{b}{0}\mm{c}{1}}_{\epsilon}|
\underbrace{\mm{a}{0}\mm{b}{1}\mm{c}{1}}_{\epsilon}|
\underbrace{\mm{a}{1}\mm{b}{1}\mm{c}{1}}_{\epsilon}
\\
\rev&\underbrace{\mm{a}{1}\mm{b}{1}\mm{c}{1}}_{\epsilon}|
\underbrace{\mm{a}{0}\mm{b}{1}\mm{c}{1}}_{\epsilon}|
\underbrace{\mm{a}{1}\mm{b}{0}\mm{c}{1}}_{\epsilon}|
\underbrace{\mm{a}{0}\mm{b}{0}\mm{c}{1}}_{c}|
\underbrace{\mm{a}{1}\mm{b}{1}\mm{c}{0}}_{\epsilon}|
\underbrace{\mm{a}{0}\mm{b}{1}\mm{c}{0}}_{b}|
\underbrace{\mm{a}{1}\mm{b}{0}\mm{c}{0}}_{a}|
\underbrace{\mm{a}{0}\mm{b}{0}\mm{c}{0}}_{\epsilon}
\end{array}
$$
\renewcommand{\arraystretch}{1}
\end{example}

Several other examples are given in Section~\ref{sec:examples}.

\begin{lemma}\label{lem:regdomlex}
    For all $f\in\LEX$, its domain $\dom(f)$ is regular.
\end{lemma}
\begin{proof}
    Observe that any \LEX transduction $f:\Sigma^*\partialf
    \Gamma^*$ is equal to $\maplex_{\ordera_1}\ (\maplex_{\ordera_2} \dots\ (\maplex_{\ordera_n}\ s)\dots)$
    for some $n\geq 0$, some ordered alphabets $(\ordera_i=(B_i,\order_i))_i$ and some simple
    transduction $s : (\Sigma\times B_1\times \dots\times
    B_n)^*\partialf \Gamma^*$. Then, $f$ is defined on
    $u\in\Sigma^*$ iff for all $1\leq i\leq n$ and all $b_i\in
    B_i^{|u|}$, $s(u\otimes b_1\otimes\dots
    \otimes b_n)$ is defined. Now, observe that $\dom(f)$ is the
    complement of the $\Sigma$-projection of the complement of $\dom(s)$. This entails
    the result as $\dom(s)$ is regular and regular languages are
    closed under morphisms (and complement). 
\end{proof}

\subsubsection{Presentation as automatic transductions}
\label{subsec:lexaut}


We give an alternative presentation in terms of automatic transductions. Let $k\ge 1$ be a positive integer, and
%
     $\overline \ordera =((B_1,\order_1),\ldots,(B_k,\order_k))$ be a $k$-tuple of ordered alphabets and let $B=B_1\times \cdots \times B_k$.
    We define the associated \emph{$k$-lexicographic order} for words
    of the same length over $B^*$ by $u\order_{\overline \ordera} v$ if
    $u=u_1\otimes \cdots \otimes u_k$, $v=v_1\otimes \cdots \otimes
    v_k$, and there is $i\in \set{1,\ldots, k}$ such that $u_i
    \order_{i} v_i$ and for all $j<i$, $u_j=v_j$.

\begin{definition}
   Let $\overline \ordera=((B_1,\order_1),\ldots,(B_k,\order_k))$ be a $k$-tuple of ordered alphabets, let $B=B_1\times \cdots \times B_k$ and let $\order_{\overline \ordera}$ be the associated $k$-lexicographic order.

    A \emph{$k$-lexicographic automatic transducer} over the alphabet $B$ is an automatic transducer with work alphabet $B$ such that the order is exactly $\order_{\overline \ordera}$. A transduction is said \emph{$k$-lexicographic automatic} if it can be defined
    by a $k$-lexicographic automatic transducer. We denote by $\AT_{k\text{-}\LEX}$ the class of $k$-lexicographic automatic transductions and by $\AT_\LEX$ the union of these, which we call lexicographic automatic transductions.
\end{definition}

The next proposition is rather immediate (Appendix.~\ref{app:atklex2klex}).

\begin{restatable}[$\AT_{k\text{-}\LEX} = \LEX_k$]{proposition}{atklextoklex}
\label{prop:atklex2klex}
    For all $k\geq 1$, a transduction is $k$-lexicographic iff it is $k$-lexicographic automatic.
\end{restatable}

Now, we make a connection between $k$-lexicographic automatic
transductions and $k$-lexicographic automatic $\omega$-words. As a consequence
of the latter proposition and known results from the literature of
automatic words, this connection allows us to prove that
$k$-lexicographic transductions form a strict hierarchy (Corollary~\ref{coro:stricthierarch}).

\begin{definition}
 Let $\overline \ordera=((B_1,\order_1),\ldots,(B_k,\order_k))$ be a $k$-tuple of ordered alphabets, let $B=B_1\times \cdots \times B_k$ and let $\order_{\overline \ordera}$ be the associated $k$-lexicographic order.
A \emph{$k$-lexicographic automatic presentation} of a word is an automatic presentation over the alphabet $B_1\times \cdots \times B_k$. such that the order over $B^*$ is exactly  the length monotonous extension of $\order_{\overline \ordera}$.

\end{definition}

Our definition differs slightly from the one of \cite{DBLP:journals/ita/Barany08}, where instead of considering a product of  $k$ alphabets, the author considers positions modulo $k$. There is however a simple correspondence between the two definitions which is made clear by the normal form lemma in \cite[Lemma~4.5]{DBLP:journals/ita/Barany08}.

\begin{proposition}\label{prop:lex-corr}
    For all $k\geq 1$, there is a one-to-one correspondence between $k$-lexicographic automatic presentations of $\omega$-words and total $k$-lexicographic automatic transductions from finite words over a unary alphabet to finite words. Moreover, when an automatic $\omega$-word $w$ corresponds to a function $f$, then $w=\Pi f$.
\end{proposition}
\begin{proof}The proof is the same as the proof of Proposition~\ref{prop:corr}. The constructions of the order automata  (which amounts to allowing words of different lengths or not) preserve the property of being $k$-lexicographic.
\end{proof}

As a consequence, we can show that the $\AT_{k\text{-}\LEX}$ form a strict hierarchy:
\begin{corollary}\label{coro:strictAT}
    For all $k\geq 1$, $\AT_{k\text{-}\LEX}\subsetneq \AT_{k+1\text{-}\LEX}$.
\end{corollary}
\begin{proof}
    This is a consequence of Proposition~\ref{prop:lex-corr} and the strictness of the hierarchy of $k$-lexicographic automatic $\omega$-words given in \cite[Theorem~6.1]{DBLP:journals/ita/Barany08}
\end{proof}

\begin{corollary}\label{coro:stricthierarch}
    For all $k\geq 1$, $\LEX_k\subsetneq \LEX_{k+1}$.
\end{corollary}
\begin{proof}
    From Corollary~\ref{coro:strictAT} and
    Propositions~\ref{prop:lex-corr},\ref{prop:atklex2klex}.
\end{proof}

\section{Examples}\label{sec:examples}

In this section, we provide a series of examples of lexicographic transductions.

\begin{example}[Morphisms]\label{ex:morphisms}
    Let $\phi_a : u\in\{a,b\}^*\rightarrow a^*$ be the morphism
    defined by $\phi_a(a) = a$ and $\phi_a(b) = \epsilon$. We have
    $\phi_a\in \LEX_1$.  It suffices to take $B = \{0,1\}$ with $0\order
    1$. Then, let $L_a = \pp{\Sigma}{0}^*\pp{a}{1}\pp{\Sigma}{0}^*$, and $L_\epsilon = \overline{L_a}$. Then:
    $$
    \phi_a\ =\ \maplex\ (L_a\ch a+L_{\epsilon}\ch
    \epsilon)
    $$
    More generally, if $\psi : \Sigma^*\rightarrow \Gamma^*$ is an
    arbitrary morphism, we show that $\psi\in \LEX_1$. Note that $\psi$
    may transform a single letter into several
    letters, while simple transductions output at most one letter. To
    overcome this difference, we consider a larger linearly
    ordered set. Let $M = \text{max}_{\sigma\in\Sigma}
    |\psi(\sigma)|$. If $M = 0$, then $\psi$ is the constant transduction
    which outputs $\epsilon$, so $\psi\in\LEX_0$. Otherwise, 
    let $\ordera_M=(B_M,<)$ with $B_M=\{0,1,\dots,M\}$ naturally ordered. Let $I :
    \Gamma\rightarrow 2^{\Sigma\times \mathbb{N}}$ such
    that for all $\gamma\in \Gamma$, $I(\gamma)$ is
    the set of pairs $(\sigma,i)$ such that $\psi(\sigma)[i] = \gamma$. 
    Note that for all $\gamma\in\Gamma$,
    $I(\gamma)\subseteq \Sigma\times \{1,\dots,M\}$. Define $L_\gamma$ as the set
    given by the regexp 
    $$
    \bigcup_{(\sigma,i)\in I(\gamma)} \pp{\Sigma}{0}^*\pp{\sigma}{i}\pp{\Sigma}{0}^*
    $$
and $L_\epsilon$ the complement of the union of all $L_\gamma$. Then
$$
\psi\ =\ \maplex_{\ordera_M}\ (L_\epsilon\ch \epsilon\ +\ \sum_{\gamma\in\Gamma}L_\gamma\ch \gamma).
$$
\end{example}

Using similar ideas, the latter example can be generalized to sequential transductions (see Appendix~\ref{app:seqlex}).

\begin{restatable}{lemma}{seqlex}
\label{lem:seqlex}
    $\SEQ\subseteq \LEX_1$.
\end{restatable}

\begin{example}[Domain restriction]\label{ex:domrestrict}Let $k\geq 0$.
Given $f:\Sigma^*\rightarrow \Gamma^*$ a transduction in $\LEX_k$ and $L\subseteq \Sigma^*$ a regular language, the transduction $f$ restricted to $L$, written $f_{|L}:u\mapsto f(u)$ if $u\in \dom(f)\cap L$ is in $\LEX_k$.
We show this inductively on $k$: it is clear for $f\in \LEX_0$. Assume $f=\maplex_{\lambda}\ g$ with $\lambda=(B,\order)$ and let $\pi_\Sigma:(\Sigma\times B)^*\rightarrow \Sigma^*$ be the natural projection morphism.
Then $f_{|L}=\maplex_\lambda\ g_{|\pi_\Sigma^{{-}1}(L)}$, which proves that $\LEX_k$ is closed under domain restriction.
\end{example}

\begin{example}[List of prefixes]
Let $\pref : \Sigma^*\rightarrow \Sigma^*$ such that $\pref(u) =
v_1v_2\dots v_k$ where each $v_1,\dots,v_k$ are successive prefixes of
$u$ of decreasing length. For example, $\pref(abcd) =
abcd.abc.ab.a$. We show that $\pref\in \LEX_1$. We take $\ordera=(B,\order)$ with $B =
\{0,1\}$ with $0\order 1$. Consider a fixed word $1^j$ for some $j$, and
some $i\geq 0$. Let $W_{i,j}$ be the set of words of length $i+j$ in the language
$1^*0^+1^j$. We have $|W_{i,j}| = i$ and all the words in $W_{i,j}$ are
lexicographically ordered as
$0^i1^j\order 10^{i{-}1}1^j\dots\order 1^{i{-}1}01^j$. Moreover,
lexicographically, all words in $W_{i,j}$ are smaller than those in
$W_{i,j+1}$. From this observation, we define the following simple
transduction: if the $B$-annotation of the input word $u$ is in $1^*0^+1^*$,
then the simple transduction produces the label of $u$ aligned with the
leftmost $0$, otherwise it outputs $\epsilon$. E.g., suppose the input word is
$abc$. The lexicographic enumeration together with the production
(above the arrows) is:
$$
000\xrightarrow{a} 100\xrightarrow{b} 010\xrightarrow{\epsilon} 110\xrightarrow{c}
001\xrightarrow{a} 101\xrightarrow{b} 011\xrightarrow{a} 111 \xrightarrow{\epsilon}
$$
More generally, let $L_\sigma =
\pp{\Sigma}{1}^*\pp{\sigma}{0}\pp{\Sigma}{0}^*\pp{\Sigma}{1}^*$ and 
$L_{\epsilon}$ the complement of $\bigcup_\sigma L_\sigma$. Then
$$
\pref\ =\ (\maplex_\ordera\ (L_\epsilon\ch \epsilon\ +\
\sum_{\sigma\in\Sigma}(L_\sigma\ch \sigma+L_{\epsilon}\ch \epsilon))
$$

Now, consider the transduction $\pref_\#$ which also adds a separator
between prefixes. 
For example $\pref_\#(abc) = abc\# ab\# c\#$. To account
for this separator, we take $\ordera'= (\{0,1,2\},<)$ with the natural
order. Note that the lexicographic successor of any annotation of the
form $2^i01^j$ is the annotation in $0^i1^{j+1}$. Also note that for
all $n$,  there are exactly
$n$ annotations of length $n$ in $2^*01^*$. The simple transduction we now
define produces the separator $\#$ on those annotations. On
annotations which contains only $0$ and $1$, it behaves as the simple
transduction of $\pref$. On all other annotations, it outputs $\epsilon$. 
Formally, let $L_\# = \pp{\Sigma}{2}^*\pp{\Sigma}{0}\pp{\Sigma}{1}^*$, and for all $\sigma$, $L_\sigma$ is
defined as before, and finally $L_\epsilon$ is the complement of the
union of all those languages. Then:
$$
\pref_\#\ =\ (\maplex_{\ordera'}\ (L_\#\ch\#\ +\ \sum_{\sigma\in\Sigma}
L_\sigma\ch\sigma\ +\ L_\epsilon\ch\epsilon)
$$
To conclude, we have shown that $\pref,\pref_\#\in\LEX_1$. 
\end{example}

\begin{example}[Subwords]
    Let $\sub : \Sigma^*\rightarrow \Sigma^*$ be the transduction
    which enumerates all the subwords of a word in lexicographic order
    (with rightmost significant bit). For example
   $\sub(abc) = a.b.ab.c.ac.bc.abc$.
   We show that $\sub \in \LEX_2$. We take $\ordera=(B,<)$, with $B=\{0,1\}$ and define the
following morphism $\era_0 : (\Sigma\times B)^*\rightarrow \Sigma^*$ by $\era_0(\sigma,0) = \epsilon$ and $\era_0(\sigma,1) =
\sigma$.     
    
    $$
    \sub\ =\ \maplex_\ordera\ \era_0
    $$
From Example~\ref{ex:morphisms}, morphisms are in $\LEX_1$,
so $\sub\in \LEX_2$.
\end{example}

\begin{example}[Square]\label{ex:square}
The transduction $\square$ has been defined in Ex.~\ref{ex:squareex}.
We show that
    $\square\in\LEX_2$. Let $\ordera=(B,<)$ with $B=\{0,1\}$ and let $f :
    (\Sigma\times B)^*\rightarrow (\Sigma\cup \underline{\Sigma})^*$
    such that for all $u\in \Sigma^*$ of length $n$, for all $1\leq
    i\leq n$, $f(u\otimes (0^{i{-}1}10^{n-i})) = \textsf{under}_i(u)$,
    and for $b\not\in 0^*10^*$, $f(u\otimes b)=\epsilon$.
    It holds that $\square = \maplex_\ordera\ f$, because
    $0^{i{-}1}10^{n-i}< 0^{j{-}1}10^{n-j}$ for all $i<j$.  
   It remains to show that $f\in \LEX_1$. It is because $f =
   \maplex_{\ordera}\ g$ for $g : (\Sigma\times B^2)^*\rightarrow
   \Sigma^*$ the following simple transduction: for all $u\otimes
   b_1\otimes b_2\in (\Sigma\times B^2)^*$, if $b_1\not\in 0^*10^*$ or
   $b_2\not\in 0^*10^*$,
   $g(u\otimes b_1\otimes b_2)=\epsilon$, otherwise let $i$ be the
   unique position at which  $1$ occurs in $b_1$ and $j$ the unique
   position at which a $1$ occurs in $b_2$. If $i=j$, then $g(u\otimes
   b_1\otimes b_2) = \underline{u[j]}$, otherwise $g(u\otimes
   b_1\otimes b_2) = u[j]$. Since those properties are regular, $g$ is
   a simple transduction. It is illustrated on Fig.~\ref{fig:square}.

\begin{figure*}[t]
$$
\underbrace{\mmm{a}{0}{0}\mmm{b}{0}{0}}_{\epsilon}\sep
\underbrace{\mmm{a}{0}{1}\mmm{b}{0}{0}}_{\epsilon}\sep
\underbrace{\mmm{a}{0}{0}\mmm{b}{0}{1}}_{\epsilon}\sep
\underbrace{\mmm{a}{0}{1}\mmm{b}{0}{1}}_{\epsilon}\sep
\underbrace{\mmm{a}{1}{0}\mmm{b}{0}{0}}_{\epsilon}\sep
\underbrace{\mmm{a}{1}{1}\mmm{b}{0}{0}}_{\underline{a}}\sep
\underbrace{\mmm{a}{1}{0}\mmm{b}{0}{1}}_{b}\sep
\underbrace{\mmm{a}{1}{1}\mmm{b}{0}{1}}_{\epsilon}\sep
\underbrace{\mmm{a}{0}{0}\mmm{b}{1}{0}}_{\epsilon}\sep
\underbrace{\mmm{a}{0}{1}\mmm{b}{1}{0}}_{a}\sep
\underbrace{\mmm{a}{0}{0}\mmm{b}{1}{1}}_{\underline{b}}\sep
\underbrace{\mmm{a}{0}{1}\mmm{b}{1}{1}}_{\epsilon}\sep
\underbrace{\mmm{a}{1}{0}\mmm{b}{1}{0}}_{\epsilon}\sep
\underbrace{\mmm{a}{1}{1}\mmm{b}{1}{0}}_{\epsilon}\sep
\underbrace{\mmm{a}{1}{0}\mmm{b}{1}{1}}_{\epsilon}\sep
\underbrace{\mmm{a}{1}{1}\mmm{b}{1}{1}}_{\epsilon}$$
   \caption{Equality $\square = (\maplex_\ordera\ \maplex_\ordera\ g)$ illustrated on input $ab$, with the results of applying $g$ underneath.\label{fig:square}}
   \end{figure*}
\end{example}

\section{Nested marble transducers}
\label{sec:nmt}

We introduce in this section a transducer model, called \emph{nested
  marble transducers}, and show in
Section~\ref{sec:equivalence} that the class of transductions it recognizes is
exactly the class of lexicographic transductions. Nested marble
transducers generalize marble transducers~\cite{DBLP:journals/actaC/EngelfrietHB99,DBLP:conf/mfcs/Doueneau-TabotF20}. A marble
transducer belongs to the family of transducers with an unbounded
number of pebbles (of finitely many colours), with the following
restriction: whenever it moves left, it has to drop a pebble, and
whenever it moves right, it has to lift a pebble. The term
\emph{marble} is meant to emphasize this restriction. 
A nested marble transducer of level $k\geq 1$ behaves like a marble transducer which can call, when reading the leftmarker $\vdash$, a nested marble transducer of level $k{-}1$. A nested marble transducer of level $0$ is what we call a simple transducer. It is just a DFA with an output function on its accepting states, so it realizes a transduction whose range is finite. 

\begin{definition}[Simple transducers]\label{def:simpletrans}
Let $\Sigma,\Gamma$ be finite sets (not necessarily disjoint). 
A $(\Sigma,\Gamma)$-simple transducer is a pair $T = (A,\out)$ where $A = (Q, q_0, F, \delta:Q\times (\Sigma\cup\{\vdash,\dashv\})\partialf Q)$ is a DFA and $\out : F\rightarrow \Gamma^{\le 1}$ is a total function. 
\end{definition}

We define two semantics for $T$, an operational semantics $f_T^{op} : Q\times \Sigma^*\partialf \Gamma^*\times F$ which takes as input a word and also a state from which the computation starts, and returns a word and the state reached when the computation ends, if it is accepting. Otherwise $f_T^{op}$ is not defined. Formally, 
$f_T^{op}(q,u)$ is defined for all $u$ such that $q\xrightarrow{{\vdash} u{\dashv}}_A q_f$ for some $q_f\in F$, by $f_T^{op}(q,u) = (\out(q_f),q_f)$.

From the operational semantics, we also define the transduction $f_T : \Sigma^*\partialf \Gamma^*$ recognized by $T$ as the transduction which applies the operational semantics from the initial state, and projects away the final state, \ie $f_T(u) = \pi_1(f_T^{op}(q_0,u))$, where $\pi_1$ is the projection on the first component.

\begin{definition}[Nested marble transducers from $\Sigma$ to $\Gamma$]
A $(0,\Sigma,\Gamma)$-nested marble transducer is a $(\Sigma,\Gamma)$-simple transducer. For $k\geq 1$, a $(k,\Sigma,\Gamma)$-nested marble transducer is a tuple $T =(\Sigma,\Gamma,C,c_0,Q_T,q_0,F_T,\delta,\deltac,\delta_r,\out,T')$ where:
\begin{itemize}
\item $C$ is a finite set of (marble) colors, $c_0$ is an initial color;
\item $Q_T$ is a finite set of states, $q_0$ is an initial state, and $F_T$ a set of accepting states;
\item $T'$ is a $(k{-}1,\Sigma\times C,\Gamma)$-nested marble transducer with set of states $Q_{T'}$ and set of accepting states $F_{T'}$;
\item $\delta: Q_T\times (\Sigma\cup\{\dashv\})\times C\rightarrow (C\cup \{\perp\})\times Q_T$ is a transition function;
\item $\deltac:Q_T\times C\rightarrow Q_{T'}$ is a call function;
\item $\deltar : Q_T\times C\times F_{T'}\rightarrow Q_T$ is a return function;
\item $\out : \dom(\delta)\rightarrow \Gamma^*$ is an output function.
\end{itemize}
We use $(k,\Sigma,\Gamma)$-\NMT (or just $k$-\NMT if $\Sigma$,$\Gamma$
are clear from the context) as a shortcut for
$(k,\Sigma,\Gamma)$-nested marble transducer. $T'$ is the
\emph{assistant} \NMT and $k$ the \emph{level} of $T$. Finally, we often say marble instead of marble colour. 
\end{definition}

We now define the semantics informally. The reading head of $T$ is initially placed on the rightmost position labeled $\dashv$, marked with a marble of color $c_0$, in state $q_0$. Transitions work as follows: suppose the current state is $q$ and the reading head is on some position $i$ labeled by $\sigma\in\Sigma\cup\{\dashv,\vdash\}$ and by some marble of color $c\in C$. Whatever transition in $\delta$ can be applied, some output word is produced by $T$ according to $\out$. Then there are three cases:
\begin{enumerate}
    \item if $\sigma\in \Sigma\cup\{\dashv\}$ and $\delta(q, \sigma, c) = (c', q')$ where $c'\in C$, then the reading head moves to position $(i{-}1)$ in state $q'$ and a marble of color $c'$ is placed (on position $i{-}1$); 
    \item if $\sigma \in\Sigma$ and $\delta(q,\sigma, c) = (\perp,q')$, then $T$ lifts the current marble and moves its reading head to position $i+1$ in state $q'$;
    \item if $\sigma =\ \vdash$ then $T$ calls $T'$ initialized with state $\deltac(q,c)$, on the input word annotated with marbles. When $T'$ finishes its computation in some accepting state $q'$, $T$ lifts marble $c$, moves its reading head to position $1$ and continues its computation from state $\deltar(q,c,q')$.
\end{enumerate}

The (operational) semantics of $T$ is a function $f_T^{op} : Q_T\times \Sigma^*\rightarrow \Gamma^*\times F_T$, that we define inductively. The case $k=0$ has been defined after Definition~\ref{def:simpletrans}. If $k\geq 1$ and $T = (\Sigma,\Gamma,C,c_0,Q_T,q_0,F_T,\delta,\deltac,\deltar,\out,T')$ then we assume $f_{T'}^{op} : Q_{T'}\times (\Sigma\times C)^*\rightarrow \Gamma^*\times F_{T'}$ to be defined inductively. Let us now define $f_T^{op}$.
A \emph{configuration} of $T$ over a word $u\in\Sigma^*$ is a triple $(q,i,v)$ such that $q$ is the current state, $i\in \pos(u)\cup \{0,n+1\}$ is the current position (where $n=|u|$), and $v\in C^*$ is an annotation of the suffix $({\vdash} u{\dashv})[i{:}n{+}1]$. We define a labeled successor relation $(q,i,cv)\xrightarrow{w}_T (q',i',v')$, between any two configurations where $c\in C$, labeled by $w\in\Gamma^*$, whenever one of the following cases hold:
\begin{enumerate}
    \item $1\leq i\leq n+1$, $\delta(q,c)=(c', q')$, $i' = i{-}1$, $v' = c'cv$ and $w = \out(q,c)$;
    \item $1\leq i\leq n$, $\delta(q,c) = (\perp, q')$, $i' = i+1$, $v' = v$ and $w = \out(q,c)$;
    \item $i=0$, $f_{T'}^{op}(\deltac(q,c),({\vdash}u{\dashv})\otimes cv) = (w,p)$, $q'=\deltar(q,c,p)$, $i'=1$ and $v' =v$.
\end{enumerate}

The function $f_T^{op}:Q_T\times \Sigma^*\partialf \Gamma^*\times F_T$ recognized by $T$ is defined, for all $q\in Q_T$ and all $u\in\Sigma^*$ such that there exists a sequence of configurations over $u$:
$$
\nu_0 = (q, n+1, c_0) \xrightarrow{w_1}_T \nu_1\xrightarrow{w_2}_T \nu_3\dots \nu_{k{-}1}\xrightarrow{w_{k}}_T \nu_{k}
$$
where the state $q_f$ of $\nu_{k}$ is accepting (\ie in $F$) and the states of configurations $\nu_i$, $i< k$, are non-accepting, by $f_T^{op}(q,u) = (w_1\dots w_k, q_f)$.

The transduction $f_T : \Sigma^*\partialf \Gamma^*$ recognized
by $T$ is defined as the projection of $f_T^{op}(q_0,u)$ on $\Sigma$
and $\Gamma$, \ie if $f_T^{op}(q_0,u) = (v,q_f)$ then $f_T(u)=v$. We
denote by $\NMT$ the class of transductions recognizable by some
$(k,\Sigma,\Gamma){-}\NMT$. The local size of an $\NMT$ is the number of its
transitions, states and marbles. Its size is its local size
plus the size the $\NMT$ of lower level it calls. We define the 
number of (resp. local number of) states/marbles/transitions similarly.

The following result states that \NMT are closed under post-composition with $\SEQ$.
To prove it, we strongly rely on the ability to pass state information through
mappings $\deltac$ and $\deltar$ to adapt a classical product construction of automata.

\begin{restatable}[$\SEQ\circ \NMT \subseteq \NMT$]{lemma}{closureseq}
\label{lem:closureseq}
    For all $k\geq 0$, all $(k,\Sigma,\Gamma)$-\NMT $T$ and all sequential transducer $S$ over $\Gamma,\Lambda$, one can construct, in polynomial time, a $(\text{max}(k,1),\Sigma,\Lambda)$-\NMT $T'$ such that $f_{T'} = f_S\circ f_T$. 
\end{restatable}

\subsubsection{State-passing free nested marble transducers} In the definition of \NMT, there are two explicit forms of information-passing: state information can be passed from level $k$ to level $k{-}1$ through the function $\deltac$, and state information can be passed from level $k{-}1$ to level $k$ via the function $\deltar$. 
In addition, there is an implicit one through the domain of assistant transducers: indeed, 
the definition of the semantics requires that all calls to assistant transducers do accept,
hence  the assistant transducer can influence the master transducer by rejecting a word.
In this subsection, we prove that information-passing can be removed while preserving the computational power of $k$-\NMT, however at the cost of increasing the size by a tower of exponentials of height $k$.
While state-passing was useful to prove the closure under post-composition
with sequential transductions (Lemma~\ref{lem:closureseq}), it will be
more convenient to consider state-passing free nested marble
transducers in the sequel, in particular to prove that \NMT
recognize lexicographic transductions (Section~\ref{sec:equivalence}).
\begin{definition}
A \emph{state-passing free $(k,\Sigma,\Gamma)$-nested marble transducer} ($(k,\Sigma,\Gamma)$-\NMTSPF for short), is either a simple transducer if $k=0$, or, if $k>0$, a 
$(k,\Sigma,\Gamma)$-\NMT $T =(\Sigma,\Gamma,C,c_0,Q_T,q_0,F,\delta,\deltac,\deltar,\out,T')$ such that 
\begin{enumerate}
    \item $T'$ is a $(k{-}1,\Sigma\times C,\Gamma)$-\NMTSPF with set of states $Q_{T'}$ and initial state $q'_0$
    \item $\deltac(q,c)=q'_0$ for all $q\in Q_T$ and $c\in C$
    \item $\deltar(q, c,q') = q$ for all $q\in Q_T,c\in C$ and $q'\in Q_{T'}$
    \item calls to the assistant transducer $T'$ always accept.
\end{enumerate}
Since the functions $\deltac$ and $\deltar$ play no role, we often omit them in the tuple denoting $T$.
\end{definition}

\begin{restatable}[State-passing removal]{theorem}{spremoval}
\label{thm:sp-removal}
    For all $k$-\NMT $T$, there exists an equivalent $k$-\NMTSPF $T'$ whose size is $k{-}\textsc{EXP}$ in that of $T$.
    This non-elementary blow-up is unavoidable.
\end{restatable}

Before proving this result, we show a property on domains of \NMT.
A \emph{nested marble automaton} $A$ of level $k$ is a nested marble
transducer $T$ of level $k$ whose output function $\out$ is the
constant function which always returns $\epsilon$. The language of $A$
is defined as $L(A) = \dom(f_T)$. 

\begin{restatable}{lemma}{regaut}\label{lem:regaut}
    A language is recognizable by a nested marble automaton of level
    $k$ and $n$ states iff it is regular iff it is recognizable by a finite automaton
    of size in $k$-EXP($n$). This non-elementary blow-up is unavoidable.
\end{restatable}

\begin{proof}[Proof sketch]
    It is clear that any regular language is the domain of some simple
    transduction. Conversely, 
    let $A$ be a nested marble automaton of level $k$. If $k=0$, it
    is obvious.
    
    If $k\ge 1$, let $A =
    (\Sigma,C,c_0,Q_A,q_0,F,\delta,\deltac,\deltar,A')$ where
    $A'$ has level $k{-}1$ (in the tuple, we have omitted the output
    alphabet and function, since they play no role). 
    By induction hypothesis, for all pairs $(q_c,q_r)$ of states of $A'$,
    the language of words accepted by $A'$ by a run starting in
    $q_c$ and ending in $q_r$ is regular, and can be described by some finite automaton
    $D_{q_c,q_r}$ of size in $(k{-}1)$-EXP($n$).

    We turn $A$ into a marble
    automaton $B$ (of level $1$) such that $L(B) = L(A)$. 
    The marbles of $B$ are enriched with information on the states
    of automata $D_{q_c,q_r}$, for all pairs $(q_c,q_r)$, ensuring
    that $B$ knows the state of all the automata $D_{q_c,q_r}$ after reading the
    current suffix. Thanks to this information, $B$
    can simulate $A$ and whenever $A$ calls $A'$ with some initial
    state $q_c$, instead, $B$ knows, if it exists,
    the state $q_r$ of $A'$ such that
    the current marked input is accepted by $A'_{q_c,q_r}$. If such a
    state exists, then it is unique as $A'$ is deterministic, and $B$
    can bypass calling $A'$ and directly apply its return
    transition. If such a state does not exist, $B$ stops. 
    
    The result follows as marble automata are
    known to recognize regular languages. The conversion of a marble
    automaton into a finite automaton is exponential both in the
    number of states and number of marbles (see
    \eg\cite{DBLP:journals/actaC/EngelfrietHB99,DBLP:conf/mfcs/Doueneau-TabotF20},
    as well as Theorem 5.4 of~\cite{DBLP:journals/iandc/KamimuraS81} for a
    detailed construction), yielding a tower of $k$-exponential inductively.  
    
    It can be shown that this non-elementary blowup is not avoidable, because first-order sentences on word structures (with one
    successor) with quantifier rank $r$, can be converted in an
    exponentially bigger nested marble automaton, while
    it is known that such sentences can be converted into an
    equivalent finite automaton of unavoidable size a tower of exponential of
    height $r$~\cite{DBLP:journals/apal/FrickG04}. Details can be found in Appendix~\ref{app:lemregaut}.
\end{proof}


  \begin{corollary}\label{coro:regdomnmt}
      Transductions recognized by nested marble transducers have regular domains.
  \end{corollary}

We are now ready to sketch the proof of Theorem~\ref{thm:sp-removal} (see Appendix~\ref{app:thm-sp-removal}
for details).
\begin{proof}[Proof sketch of Theorem~\ref{thm:sp-removal}]
There are two kinds of state-passing, through functions $\deltac$ and $\deltar$.
We deal with them separately. First, removal of $\deltac$ can be done by enriching marbles with the current state,
so as to pass this information to the assistant transducer.
Removal of $\deltar$ is more involved, but can be done by induction using a technique similar
to the one used to prove Lemma~\ref{lem:regaut}. By induction hypothesis, the assistant transducer
can be replaced by an equivalent state-passing free \NMT. In addition, its domain is regular thanks to Lemma~\ref{lem:regaut}. Hence, one can enrich the marbles so as to precompute the final state
reached by the assistant transducer, and in turn simulate the function $\deltar$. This also allows to ensures that 
all calls to the assistant transducer do accept.

Last, we justify the fact that the non-elementary blow-up is unavoidable. It is because the domain of any state-passing free \NMT $S$ is recognizable by a finite automaton of exponential size. Indeed, the calls to assistant transducers always terminate, so the domain of $S$
does not depend on assistant transducers, hence can be described by a marble automaton, hence by a finite
automaton of size exponential in the number of local states and marbles of $S$. Thus, the existence of an elementary construction for state-passing removal
would contradict the non-elementary blow-up stated in Lemma~\ref{lem:regaut}.
\end{proof}

\section{Nested marble transducers capture the class of
  lexicographic transductions}
  \label{sec:equivalence}

In this section, we prove (Theorems~\ref{thm:lex2marble} and~\ref{thm:marble2lex}) that a transduction is
recognizable by a nested marble transducer of level $k$ iff it is
$k$-lexicographic. 

\subsubsection{Lexicographic transductions are recognized by nested marble transducers}

Consider some transduction $f\in LEX_k$. Then $f = \maplex_{\ordera_1} (\maplex_{\ordera_2}\dots
\maplex_{\ordera_k}\ s)\dots )$ for some $\ordera_i=(B_i,\order_i)$ and $s$ a simple transduction. We
call $B_1,\dots,B_k$ \emph{the} ordered alphabets of $f$ and $s$
\emph{the} simple transduction of $f$.

Given an ordered alphabet $(B,\order)$, one can enumerate
the annotations of a word according to the lexicographic extension
using a marble automaton. By induction, we show:

\begin{restatable}[$\LEX\subseteq \NMT$]{theorem}{lextomarble}
\label{thm:lex2marble}
    Any transduction $f\in \LEX_k$ is
    recognizable by some \NMT $T_f$ of level $k$. If
    $B_1,\dots,B_k$ are the ordered alphabets of $f$ and $s$ its
    simple transduction, represented by a sequential transducer with $m$
    states, then $T_f$ has $O(k+m)$ states and $\sum_{i=1}^k |B_i|$
    marbles. 
\end{restatable}

\subsubsection{Nested marble transducers recognize \LEX transductions}

Conversely, we prove that the transductions recognized by \NMTSPF are lexicographic. 

\begin{restatable}[$\NMTSPF\subseteq \LEX$]{theorem}{marbletolex}
\label{thm:marble2lex}
    Any transduction $f$ recognizable by
    some $\NMTSPF$ $T$ of level $k$ is $k$-lexicographic. 
\end{restatable}

\begin{proof}[Proof sketch]
The proof is rather involved. We provide high level arguments, but full details can be found in Appendix~\ref{app:marble2lex}. The result is shown by induction on $k$. It is trivial for $k=0$. For $k>0$, the main idea is to see the sequence of successive configurations of $T$ on some input as a lexicographic enumeration. This is possible due to the stack discipline of marbles. In particular, by extending the marble alphabet with sufficient information, it is possible to define a total order on marbles
such that the sequence of successive configurations of $T$,
extended with this information, forms a lexicographically increasing chain.

More precisely, since lexicographic functions only produce their
output at the \emph{bottom level}, \ie at the level of simple
functions, we first turn $T$ into an equivalent \emph{bottom
  producing} transducer, where only its simple transducer is allowed
to produce outputs. Then, we observe that the behaviour of a nested marble transducer is inherently
lexicographic. To see this, let ${\vdash}u{\dashv}$ be some input
word such that $u\in\dom(f)$. Consider a configuration of $T$ where
the reading head is on ${\vdash}$, \ie a configuration of the form
$\nu = (q,0,v)\in Q\times \pos({\vdash}u{\dashv})\times C^*$, reachable from
the initial configuration. For any input position $i$, denote by
$\textsf{last}_\nu(i)$ the last state in which $T$ was at position $i$
before reaching $\nu$.

Now, consider two configurations $\nu_1 = (q_1,0,v_1),\nu_2 = (q_2,0,v_2)$ on
${\vdash}u{\dashv}$ reachable from the initial configuration. It can
be proved that $\nu_1\rightarrow_T^*\nu_2$ iff there exists a position
$i$ such that $(1)$ for
all positions $j\geq i$,
$\textsf{last}_{\nu_1}(j)=\textsf{last}_{\nu_2}(j)$ and $v_1[j] =
v_2[j]$, and $(2)$ the pair $(\textsf{last}_{\nu_1},\textsf{last}_{\nu_2})$ is a \emph{right-to-right traversal} of $T$ for position $i$ and marble $c = v_1[i]$, \ie a pair of states $(p_1,p_2)$ such that for some $\alpha$, $T$ can go from configuration $(p_1,i,c\alpha)$ to configuration $(p_2,i,c\alpha)$ without visiting any position to the right of $i$ (in particular, being a right-to-right traversal does not depend on $\alpha$). This is illustrated in Fig.~\ref{fig:illustrateorder} in Appendix.

Based on this observation, $f$ is expressed as $\maplex_\ordera\ f'$
with $\ordera = (B,\order_B)$ where $B = C\times Q\times 2^{Q\times
  Q}$ and $f'$ checks that its input, of the form $c_0c_1\dots \otimes
p_0p_1\dots \otimes X_0X_1\dots$,  is valid,
in the sense that $(p_0,0,c_0c_1\dots)$ is a configuration $\nu$ of
$T$ reachable from the initial configuration, $p_i = \textsf{last}_\nu(i)$ for all $i$,
and $X_i$ is the set of right-to-right traversals at position $i$. It
is proved that valid inputs form a regular language. 
If its input is valid, $f'$ outputs $f_R({\vdash}u{\dashv}\otimes
c_0c_1\dots)$, for $R$ the assistant transducer of $T$, otherwise $f'$
outputs $\epsilon$. $f'$ is shown to
be lexicographic by induction hypothesis. Finally, $\order_B$ is defined by
$(c_1,q_1,X_1)\order_B (c_2,q_2,X_2)$ if $c_1=c_2$ and $X_1=X_2$ and
$(q_1,q_2)\in X_1$, otherwise the order is arbitrary. 
\end{proof}

The following theorem summarizes the characterizations of
lexicographic transductions proved so far:

\begin{theorem}[Characterizations of lexicographic transductions]\label{thm:main} Let
  $f : \Sigma^*\rightharpoonup \Gamma^*$ be a transduction and $k\geq 1$. The
  following statements are equivalent:
  \begin{enumerate}[$(1)$]
    \item $f$ is $k$-lexicographic.
    \item $f$ is $k$-lexicographic automatic.
    \item $f$ is recognizable by a $k$-nested marble transducer.
    \item $f$ is recognizable by a state-passing free $k$-nested marble
      transducer.
  \end{enumerate}
\end{theorem}

\begin{proof}
    $(2)\xleftrightarrow{\text{Prop.~\ref{prop:atklex2klex}}} (1)\xrightarrow{\text{Thm~\ref{thm:lex2marble}}} (3)\xrightarrow{\text{Thm.~\ref{thm:sp-removal}}}
    (4)\xrightarrow{\text{Thm.~\ref{thm:marble2lex}}}(1)$.
\end{proof}

\section{Expressiveness and closure properties of lexicographic transductions}
\label{sec:ppties}

In this section, we prove that $\LEX$ contains all the polyregular
transductions~\cite{DBLP:polyreg}, and all the transductions recognizable by (copyful)
streaming string transducers~\cite{DBLP:journals/fuin/FiliotR21,DBLP:conf/fsttcs/AlurC10}. We also show that $\LEX$ is
closed by postcomposition under any polyregular transduction, and
closed by precomposition under any rational transduction.
We start by
showing that lexicographic transductions preserve regular languages under
inverse image. 

\begin{proposition}
    Transductions in \LEX are regularity preserving.
\end{proposition}

\begin{proof}
    It is an immediate corollary of the inclusion $\LEX\subseteq \NMT$ (Theorem~\ref{thm:main}), that \NMT are closed by
postcomposition by a sequential transduction (Lemma~\ref{lem:closureseq}), and
that \NMT have regular domains (Corollary~\ref{coro:regdomnmt}). 
\end{proof}

We show that \LEX subsumes both \SST and \PREG.
More precisely that any transduction recognizable by a (copyful) streaming string transducer (\SST) is $1$-lexicographic. We do not give the definition of \SST and refer the reader to~\cite{DBLP:conf/mfcs/Doueneau-TabotF20} for more details.
We also show that \NMT  of level $k$ subsume $k$-pebble transducers. Again we don't give precise definitions of pebble transducers and refer the reader to \cite{DBLP:polyreg}.

\begin{theorem}[$\SST$ and \PREG in \LEX]
The following hold:
\begin{itemize}
    \item   $\SST=\LEX_1$,
    \item A transduction definable by a $k$-pebble transducer is in $\LEX_k$.
    In particular $\PREG\subseteq \LEX$
\end{itemize}

\end{theorem}

\begin{proof}
    It is already known that marble transducers (\ie nested marble
    transducers of level $1$) capture exactly the class of
    \SST-recognizable transductions~\cite{DBLP:conf/mfcs/Doueneau-TabotF20}. The result then follows
    from Theorem~\ref{thm:main}.

    A single pebble can be simulated by one level of marbles, with colors $\set{0,1}$ with the restriction that at most one marble can have color $1$ per level.
\end{proof}

We now prove that the class of lexicographic transductions is closed under
postcomposition by a polyregular transduction.

\begin{theorem}\label{thm:closurepolyreg}
$\textsf{PolyReg}\circ \LEX\subseteq \LEX$
\end{theorem}

\begin{proof}
    Any polyregular transduction can be expressed as a composition of
    sequential transductions, \square, \map and \rev~\cite{DBLP:polyreg}. We show that
    $\LEX$ is closed by postcomposition by these transductions, which
    yields the result. For sequential transductions, it has been shown in
    Lemma~\ref{lem:closureseq}. It remains to show closure under
    \square, \map and \rev.

\vspace{2mm}
    \emph{Closure under \map} 
    Let $f : \Sigma^*\rightarrow \Gamma^*$ be a
    $k$-lexicographic transduction. We show that the transduction $\map\ f : \Sigma_|^*\rightarrow \Sigma_|^*$ (where $\map$ is defined in Ex.~\ref{ex:map}) is $k$-lexicographic. Suppose that $f$ is given as an
    $\NMTSPF$ $T_f$ of level $k$. We sketch the construction of an $\NMTSPF$ $T$ of
    level $k+1$ recognizing $\map\ f$. The transducer $T$ only
    needs to mark chunks of the form $|u_i|$,
    ${\vdash} u_i|$ or $|u_i{\dashv}$, for increasing values of
    $i$. It can do so by using three marbles, $\bot$,
    $\triangleright$ and $\triangleleft$: $\triangleright$ is used to
    mark the left separator of the chunk (either $|$ or $\vdash$),
    $\triangleleft$ is used to mark the right separator (either $|$ or
    $\dashv$) and $\bot$ is used to mark any other position. It is not
    difficult to construct $T$ such that it enumerates all the
   markings of chunks of its inputs in order from left to right. 
    Once a chunk has been correctly marked by $T$, $T$ moves to the initial
    position of the whole input word, and call a modified version of
    $T_f$, which we call $T'_f$. The transducer $T'_f$ behaves as
    $T_f$ but on the portion of its input in between the position
    marked $\triangleright$ and $\triangleleft$. Accordingly,
    transitions of $T_f$ are modified such that $\triangleright$ is
    interpreted as $\vdash$, and $\triangleleft$ as $\dashv$. The
    subtransducer of $T_f$ (and recursively the subtransducer of its
    subtransducer), also has to be modified to be simulated only on
    the marked chunk. When $T'_f$ reaches an accepting state, the
    control goes back to $T$ which then marks the next chunk. With
    a slightly more technical construction (but similar), it is possible to
    both combine $T$ and $T'_f$ at the same level, so that $\map\ f$ can be shown to be
    $k$-lexicographic instead of $(k+1)$-lexicographic. 
    
\vspace{2mm}
    \emph{Closure under \rev}  We prove by induction on $n$
    that for all $k\geq 0$, if $f\in\LEX_k$, then $\rev\ f\in
    \LEX_k$. For $k=0$, it is easy to see that simple transductions are
    closed under reverse since words of length at most one are their own reverse. If $k>0$, there exist
    $\ordera = (B,\order)$ an ordered alphabet and $(g\ :\ (\Sigma\times
    B)^*\partialf \Gamma^*)\in \LEX_{k{-}1}$ such that $f\ =\
    \maplex_\ordera\ g$.
    
    Let $\ordera^{{-}1} = (B,\order^{{-}1})$ where $x\ \order_B^{{-}1}\ y$ iff $y\order_B x$. We prove that $\rev\ f =
    \maplex_{\ordera^{{-}1}}\ (\rev\ g)$.
    First, the lexicographic extension of $\order^{{-}1}$ is precisely the reverse order of $\order$: for $u,v\in B^n$, if $u\order v$ then there is $i\leq n$ such that $u[i]\order v[i]$ and $u[j]\order v[j]$ for all $i<j\leq n$. This means hat $v\order^{{-}1} u$.
     Then,
    take an input word $u\in\Sigma^*$ of length $n$, and let
    $v_1,\dots,v_k\in B^*$ of length $n$, we have
    $$
    v_1\order^{{-}1}_B v_2\dots v_{k{-}1}\order_B^{{-}1}
    v_k\text{ iff } v_k\order_B v_{k{-}1}\dots v_{2}\order_B v_1
    $$

 Suppose now that $v_1,\dots,v_k$ is the lexicographic enumeration of all
 $B$-words of length $n$ according to $\order_B^{{-}1}$, then:
 $$
 \begin{array}{lllllllll}
   \rev\ f\ u  & = & \rev\ (\maplex_\ordera\ g)\ u \\
               & = & \rev (g(v_k)\dots g(v_1)) \\
               & = & (\rev\ g\ v_1)(\rev\ g\ v_2)\dots (\rev\ g\ v_k) \\
   & = & \maplex_{\ordera^{{-}1}}\ (\rev\ g)\ u
 \end{array}
 $$

\vspace{2mm}
\emph{Closure under \square} We refer the reader to Example~\ref{ex:square}
for a definition of the transduction $\square$. Let $f\in \LEX$. If $f$ is simple, then it is
immediate that $(\square\ f)\in \SEQ$ and the result follows from
Lemma~\ref{lem:seqlex}. Otherwise, assume
$f\in\LEX_{k}$ for $k>0$, we show that $(\square\ f)\in \LEX_{2k}$. Then
$f\ = \maplex_{\ordera_1}\ \maplex_{\ordera_2}\dots \maplex_{\ordera_k}\ s
$ where $\ordera_i = (B_i,\order_i)$, for ${1\leq i\leq  k}$, are ordered alphabets and 
$s:(\Sigma\times B_1\times\dots\times B_k)^*\partialf
\Sigma^{\leq 1}$ is simple. Now, define the transduction
$\underline{s}:(\Sigma\times (\Pi_{i=1}^k B_i)\times (\Pi_{i=1}^k B_i))^*
                                                                  \partialf
                                                                                    (\Sigma\cup\underline{\Sigma})^{\leq
                                                                                      1}$
                                                                                    such
                                                                                    that
                                                                                    for
                                                                                    all
                                                                                    $u\in\Sigma^*$
                                                                                    and
                                                                                    all
                                                                                    $\overline{b},\overline{c}\in
(B_1\times\dots\times B_k)^*$, 
$$
\begin{array}{llllllllll}

\underline{s}(u\otimes \overline{b}\otimes \overline{c}) & = &
                                                     \left\{\begin{array}{lcl}
                                                              \epsilon
                                                              &\text{if}&
                                                                          s(u\otimes
                                                                          \overline{b})=\epsilon,\text{
                                                                          else}\\
                                                               s(u\otimes
                                                               \overline{c})
                                                               &\text{if}&
                                                                           \overline{b}
                                                                           \neq\overline{c}
                                                               \\
                                                               \underline{s(u\otimes
                                                               \overline{c})}&\text{if}& \overline{b}=\overline{c}
                                                               \end{array}\right.
\end{array}
$$
where $\underline{\epsilon}=\epsilon$.  It is not difficult to see
that $\underline{s}$ is simple. Now, note that the number of times
$s(u\otimes \overline{b})\neq \epsilon$ holds is exactly $|f(u)|$. Moreover,
whenever $s(u\otimes\overline{b})\neq \epsilon$, the enumeration of
all the $\overline{c}$ produces $f(u)$, where the letter produced
when $\overline{c}=\overline{b}$ is underlined. Therefore, 
$$
\square\ f = (\maplex_{\ordera_1}\ \dots\ \maplex_{\ordera_k})^2\ \underline{s}
$$
where $.^2$ here has to be understood as a copy, \eg $(\alpha\ \beta)^2 = \alpha\ \beta\ \alpha\ \beta$ for any function $\alpha,\beta$.
\end{proof}

Finally, we show that lexicographic transductions are closed by pre-composition under any rational transduction.

\begin{restatable}[$\LEX\circ\RAT\subseteq \LEX$]{theorem}{precompose}
\label{thm:precompose}
Let $k\geq 1$. For any rational transduction $g:\Gamma^*\partialf
    \Lambda^*$ and any $k$-lexicographic transduction $f:\Lambda^*\partialf
    \Gamma^*$, $f\circ g$ is $k$-lexicographic.
\end{restatable}

\begin{proof}
    A rational transduction can be decomposed as a letter-to-letter rational transduction, followed by a morphism. Example~\ref{ex:morphisms} shows that morphisms are lexicographic. Similar ideas apply inductively to show that \LEX is closed by precomposition under morphisms and letter-to-letter rational transductions. See Appendix~\ref{app:precompose}.
    \end{proof}

\section{Discussion}

We have introduced the class of lexicographic transductions, a subclass of \MSOSI.
We have given three different characterizations: in terms of closure of simple functions by the \maplex\  operator, as lexicographic automatic transductions which can be seen as a syntactic restriction of \MSOSI, and finally as nested marble transducers.

Thanks to these characterizations we have shown that this class subsumes both \PREG and \SST. Moreover it is actually closed under post-composition by \PREG and thus is regularity preserving, which \MSOSI is not known to be.

\subsubsection{Open questions on \MSOSI}
We leave open whether \MSOSI is regularity preserving. A way to attack the problem is trying to obtain an equivalent automata-based model, as stated in Question~\ref{ques:1}.
However, as we have shown, a positive answer would entail that all automatic $\omega$-words have a decidable \MSO theory, which has been open for almost 20 years.

Since the inverse image of an \FO definable language by an \MSOSI is regular, it means in the case of words that the inverse image of any language recognized by an \emph{aperiodic} monoid is regular, thanks to the famous result of Schützenberger \cite{DBLP:journals/iandc/Schutzenberger65a}. Using a decomposition result of Krohn-Rhodes \cite{KrohnRhodes} this means that to show that \MSOSI over words are regularity preserving one only needs to show that the inverse image of language recognized by a group (even a simple group) is regular. Moreover, it can be shown \cite[Theorem~3.2]{DBLP:conf/lics/Gradel20} that \MSOSI over words satisfy the backward translation property for first-order logic with modulo quantifiers, which can say for instance \textit{there is an even number of positions $x$ such that $\phi(x)$ holds}. This logic corresponds to languages recognized by \emph{solvable} monoids, that is monoids containing only solvable groups.
Hence,  to show that word-to-word \MSOSI are \emph{not} regularity preserving, one could start by considering languages recognized by $\mathcal A_5$ the smallest non-solvable group, which is a subgroup of the permutation group of $5$ elements $\mathcal S_5$.

\subsubsection{Open questions on \LEX}
We have shown that \LEX is a rather well-behaved class, however some interesting questions remain open.
We have not shown that \MSOSI \emph{strictly} subsumes \LEX, although we suspect it does. For instance we have not been able to show that \LEX is closed under pre-composition with \REG, and we believe that it is not.

The equivalence problem is central to transducer theory, however equivalence of \PREG transductions is not known to be decidable, and since \PREG is subsumed by \LEX it is also the case for the latter.

One can decide whether a \LEX transduction is in \PREG, however we don't know whether it is decidable if a \LEX transduction is in $\LEX_1$, and more generally compute on which level of the hierarchy it sits.

\subsubsection{Possible extensions of \LEX}
While \LEX has proven to be an interesting class, the zoo of word-to-word transductions with exponential growth is relatively unknown.
We propose three possible extensions of \LEX, in increasing expressiveness, all of which turn out to be included in \MSOSI.

The class $\LEX\circ\REG$ may be an interesting class in itself and the first level $\LEX_1\circ \REG$ coincides with the rather natural class of two-way streaming string transducers.

A more general way of extending \LEX is to generalize the operation \lexenum{} to allow lexicographic orders where the significance of letters is given by an arbitrary \MSO definable order. This class subsumes $\LEX\circ \REG$ however it is not clear that it is regularity preserving.

Another possible generalization is to replace marbles with the so-called invisible pebbles of Engelfriet\cite{DBLP:journals/tcs/EngelfrietHS21} and define \emph{nested invisible pebble transducers}, where the nested levels can see the pebbles of the previous levels but not the ones of their own. The state-passing free version can be shown to be still included in \MSOSI but it is not clear that it is regularity preserving. However the version with state-passing is too expressive since it can recognize non-regular languages.

\bibliographystyle{alpha}
\bibliography{main}

\newpage

\appendices

\section{Omitted proofs of Section~\ref{sec:msosi}}

\subsubsection{Proof of Proposition~\ref{prop:back-seq}}

\backseq*

\begin{proof}[Proof sketch]
     
     (2)$\Rightarrow$ (1) Given a transduction computed by an \MSOSI, and a regular language $L$, take a deterministic automaton recognizing it, and let $g$ be computed by the Mealy machine which labels every letter by the target state of the transition. Let $L'$ denote the language of annotated words whose last letter is annotated by a final state.
     Then $g\circ f(u)\in L'$ if and only if $f(u)\in L$. Since $L'$ is first-order definable and \MSOSI are closed under \FOI, then $(g\circ f)^{{-}1}(L')=f^{{-}1}(L)$ is regular.
    (1)$\Rightarrow$ (2) Let $f:\Sigma^*\partialf \Gamma^*$ be defined by an \MSOSI of dimension $k$.  Let $u$ be a word and let  $P\subseteq \pos(u)$, we can naturally associate a word $v\in \set{0,1}^{|u|}$ representing the set $P$. For $P_1,\ldots,P_k\subseteq \pos(u)$, we use the notation $u(P_1,\ldots, P_k)$ for the word $u\otimes v_1\otimes\ldots v_k$ where $P_i$ is represented by $v_i$.

    We can define an \MSOSI realizing $f_{\text{prefix}}:\Sigma\times \set{0,1}^k\partialf \Gamma^*$, which maps $u(\overline P)$ to the prefix of $f(u)$ restricted to the $k$-tuples of subsets coming \emph{strictly} before $\overline P$.

    let us consider a Mealy machine computing a transduction $g:\Gamma^*\rightarrow\Lambda^*$ and let $L_p$ the set of words $u$ such that $u$ reaches state $p$ in the Mealy machine, which is regular. Then by assumption $f_{\text{prefix}}^{{-}1}(L_p)$ is a regular language. Thus we can define a formula $\psi_p(X_1,\ldots,X_k)$ such that $u\models \psi_p(P_1,\ldots,P_k)$ if and only if $u(\overline P)\in f_{\text{prefix}}^{{-}1}(L_p)$.
    
    If we have $\phi_\gamma(X_1,\ldots,X_k)$ defining the positions of $f(u)$ labeled by $\gamma$, then we define the positions of $g\circ f(u)$ labeled by $\alpha$ by a disjunction for all states $p$ with a transition $p\xrightarrow{\gamma|\alpha}q$: $$\bigvee_{p\xrightarrow{\gamma|\alpha}q}\psi_p(X_1,\ldots,X_k)\wedge \phi_\gamma(X_1,\ldots,X_k)$$
    (3)$\Rightarrow$ (2) Mealy machines compute particular cases of polyregular transductions.
    
    (2)$\Rightarrow$ (3) For this direction we rely on several decomposition results. The first, from \cite[Lemma~9]{BKL19} says that any polyregular transduction can be obtained as the composition of a first-order interpretation and a rational letter-to-letter transduction.
    Secondly, rational transductions are known to be compositions of sequential and co-sequential transductions \cite{DBLP:books/lib/Berstel79}. Moreover, a co-sequential function is the composition of the form $\rev\circ f\circ \rev$ with $f$ sequential, and $\rev$ is an \FOI.
\end{proof}

\section{Omitted proofs of Section~\ref{sec:lex} and Section~\ref{sec:examples}}

\subsubsection{Proof of Proposition~\ref{prop:atklex2klex}}\label{app:atklex2klex}

\atklextoklex*

\begin{proof}
    Let $f:\Sigma^*\partialf \Gamma^*$. Suppose that 
    $f = \maplex_{\ordera_1}\ \maplex_{\ordera_2} \dots\
    \maplex_{\ordera_k}\ s$ for some simple function $s$ of the form
    $L_\epsilon\ch\epsilon+\sum_{\gamma\in\Gamma} L_\gamma\ch\gamma$. 
    Then
    $\mathcal{A}_{dom}$ is defined as any automaton recognizing
    $\dom(f)$, which is regular by
    Lemma~\ref{lem:regdomlex}. $\mathcal{A}_{univ}$ is any automaton 
    recognizing $\dom(s)$, which is also regular. Finally, for all
    $\gamma\in\Gamma$, $\mathcal{A}_\gamma$ is any automaton
    recognizing $L_\gamma$. 

    Conversely, given $\mathcal{A}_{dom},\mathcal{A}_{univ}$ and
    $\mathcal{A}_\gamma$ for all $\gamma\in\Gamma$ defining $f$, we
    define a simple function $s$ such that $f = \maplex_{\ordera_1}\ \maplex_{\ordera_2} \dots\
    \maplex_{\ordera_k}\ s$. Define $L_\epsilon$ the (regular) set of
    words $(u\otimes b_1\dots\otimes b_k)\in (\Sigma\times
    B_1\times\dots\times B_k)^*$ such that $u\in
    L(\mathcal{A}_{dom})$. Then, $s = L_\epsilon\setminus(\bigcup_{\gamma\in \Gamma}L(\aut_\gamma)\ch \epsilon +
    \sum_{\gamma\in\Gamma} L(\mathcal{A}_\gamma)\ch\gamma$. 
\end{proof}

\subsubsection{Proof of Lemma~\ref{lem:seqlex}}\label{app:seqlex}

\seqlex*

\begin{proof}
    Let $T = (A,\out)$ a sequential transducer recognizing a
    transduction $f : \Sigma^*\partialf \Gamma^*$, where $A = (\Sigma,Q,q_0,F,\delta)$.  We define an ordered alphabet $\ordera=(B,\order)$ and a simple transduction $s : (\Sigma\times
    B)^*\partialf \Gamma^{\leq 1}$ such that $f = \maplex_{\ordera}\ s$.
    Let $M = \text{max}_{(q,\sigma)\in Q\times \Sigma}
    |\out(q,\sigma)|$ be the length of the longest word occurring
    on the transitions of $T$.

    We take $B = Q\times \{\bot,1,\dots,M\}$. Intuitively, the
    $Q$-component is used to annotate the input word $u$ with the run of
    $T$, if $u\in\dom(f)$, and exactly one position is annotated with
    some integer $j\neq \bot$, meant to identify the $j$-th letter of
    the output word produced by $\out$. Since not all annotations
    satisfy these constraints, we define the set of \emph{valid}
    annotated words as the set $V$ of words $v =
    (\sigma_1,(q_1,b_1))\dots (\sigma_n,(q_n,b_n))\in (\Sigma\times B)^*$
    such that $\sigma_1\dots\sigma_n\in\dom(f)$, $q_0q_1\dots q_n$ is \emph{the} accepting
    run of $A$ on $\sigma_1\dots\sigma_n$, and, there is a unique $i$
    such that $b_i\neq\bot$, and in that case,
    $b_i\in\{1,\dots,|\out(q_{i{-}1},\sigma_i)|\}$. $V$ is easily
    seen to be regular. If $v\in V$, then we let $s(v) =
    \out(q_{i{-}1},\sigma_i)[b_i]$. If $v\not\in V$ but
    $\sigma_1\dots\sigma_n\in\dom(f)$, then we let $s(v)=\epsilon$. In
    all other cases, $s(v)$ is undefined. The transduction $s$ is simple.

    It remains to order $B$ so that the lexicographic enumeration of
    valid annotations coincides with the sequential productions of outputs by
    $T$. To do so, we let $(q_1,b_1)\order (q_2,b_2)$ if $b_1< b_2$ (with
    $\bot$ being smaller than any integer), and
    if $b_1=b_2$, $q_1<_Q q_2$ for some arbitrary order $<_Q$ over $Q$.
    For example, for a sequential transducer with a single state $q$
    (both initial and final) over $\Sigma=\Gamma=\{a\}$, a loop on
    state $q$ with
    $\out(q,a)=aa$, the \emph{valid} annotations of the input word $aa$
    are lexicographically enumerated as follows:
    $$
    \pp{q_0}{1}\pp{q_0}{\bot}\order     \pp{q_0}{2}\pp{q_0}{\bot}\order
    \pp{q_0}{\bot}\pp{q_0}{1}\order
    \pp{q_0}{\bot}\pp{q_0}{2}
    $$
\end{proof}

\section{Omitted proofs of Section~\ref{sec:nmt}}
\label{app:nmt}

\subsubsection{Proof of Lemma~\ref{lem:closureseq}  }

\closureseq*

\begin{proof}
If $k=0$, $T$ is a simple transducer defined as a DFA with outputs on
accepting states. Equivalently, $T$ can be seen as a sequential
transducer which produces the empty word all the time, and produces a
(possibly) non-empty word when reading the rightmarker $\dashv$. Then,
it is well-known that sequential transducers are closed under
composition, and the result follows by Lemma~\ref{lem:seqlex}.

If $k>0$, the composition is realized by a product construction
between $T$ and $S$ that we denote $T\otimes S$. States of $T\otimes S$ are therefore pairs of states of
$T$ and $S$: $Q_{T\otimes S} = Q_T\times Q_S$. Whenever $T$ applies its output function $\out$ on
some state $q$ and marble $c$, producing a word $\out(q,c)\in\Gamma^*$ and moving to
some state $q'$, in the
product construction, from a pair of states $(q,s)$, if
$s\xrightarrow{\out(q,c)/w}_S s'$, then $T\otimes S$ produces $w$ and moves
to state $(q',s')$. Now, suppose $T$ reads the leftmarker $\vdash$ in
some state $q$ and marble $c$, and calls a $(k{-}1){-}\NMT$ $R$ initialized with state
$p=\deltac(q,c)$. In the product, if $T\otimes S$ is in state $(q,s)$ with marble $c$, then it
calls $R\otimes S$, whose construction is obtained inductively,
initialized with state $(p,s)$. When $R\otimes S$ returns in
some state $(p',s')$, then $T\otimes S$ continues its computation from
state $(\deltar(q,c,p'),s')$. The construction ensures the following invariant: for all $q\in Q_T$, all $s\in Q_S$, all $u\in\Sigma^*$, all $v\in\Gamma^*$ and all $w\in \Lambda^*$, the following holds:
$$
\begin{array}{llllllll}
    (f_T^{op}(q,u) = (v,q')\ \wedge\ f_S^{op}(s,v) =
  (w,s'))\\\qquad\qquad \Rightarrow f_{T\otimes S}^{op}((q,s), u) = (w, (q',s')).
  \end{array}
    $$

\end{proof}

\subsubsection{Proof of Lemma~\ref{lem:regaut}}
\label{app:lemregaut}

\regaut*

\begin{proof}
    It is clear that any regular language is the domain of some simple
    transduction. Conversely, 
    let $A$ be a nested marble automaton of level $k$. If $k=0$, it
    is obvious. If $k=1$, then $A$ is a marble automaton, and marble automata are
    known to recognize regular languages. The conversion of a marble
    automaton into a finite automaton is exponential both in the
    number of states and number of marbles (see
    \eg\cite{DBLP:journals/actaC/EngelfrietHB99,DBLP:conf/mfcs/Doueneau-TabotF20},
    as well as Theorem 5.4 of~\cite{DBLP:journals/iandc/KamimuraS81} for a
    detailed construction).

    If $k>1$, let $A =
    (\Sigma,C,c_0,Q_A,q_0,F,\delta,\deltac,\deltar,A')$ where
    $A'$ has level $k{-}1$ (in the tuple, we have omitted the output
    alphabet and function, since they play no role). For all states $q_c,q_r$
    of $A'$, we let $A'_{q_c,q_r}$ be the automaton $A'$ where $q_c$
    is initial, and $q_r$ is the only final state.
    By induction
    hypothesis, $L(A'_{q_c,q_r})$ is regular for all pairs $(q_c,q_r)$, recognizable by some finite
    automaton $D_{q_c,q_r}$, which we assume to be backward deterministic
    (\ie there is a single accepting state, and for any state $q$ and
    any symbol $a$, there is at most one incoming transition to $q$
    reading $a$). We turn $A$ into a marble
    automaton $B$ (of level $1$) such that $L(B) = L(A)$.
    Since $A$ processes its input from right to left, we use
    the backward-determinism of the automata $D_{q_c,q_r}$ and extra marbles such that at any
    point on the input, $B$ knows the state of all the automata $D_{q_c,q_r}$ after reading the
    current suffix from right to left. Thanks to this information, $B$
    can simulate $A$ and whenever $A$ calls $A'$ with some initial
    state $q_c$, instead, $B$ knows, if it exists,
    the state $q_r$ of $A'$ such that
    the current marked input is accepted by $A'_{q_c,q_r}$. If such a
    state exists, then it is unique as $A'$ is deterministic, and $B$
    can bypass calling $A'$ and directly apply its return
    transition. If such a state does not exist, $B$ stops. It is easy
    to construct $B$ such that it can compute this information, by
    marking the input with all the states of all automata
    $D_{q_c,q_r}$ reached on the current suffix from right to left,
    and by using backward-determinism to compute such tuples
    deterministically from right to left. The construction is
    exponential, and yields a tower of $k$-exponential inductively. 

    We construct $B$
    such that it satisfies the following invariant: for all configuration
    $\nu = (q,i,(c,\Psi)v)$ 
    of $B$, on some input word
    $u = \sigma_0\sigma_1\dots\sigma_n$, where $q\in Q_A$, if $\nu$ is
    the $j$-th configuration reachable from the
    initial configuration, then 
    \begin{enumerate}[$(i)$]
    \item $(q,i,c\pi_C(v))$ is the $j$-th configuration
    of $A$ on $u$ reachable from the initial configuration, and 
    \item 
    for all pairs $(q_c,q_r)\in Q_{A'}^2$, $\Psi(q_c,q_r)$ is the state of $D_{q_c,q_r}$ after reading the word
    $(\sigma_i\sigma_{i+1}\dots\sigma_n)\otimes (c\pi_C(v))$.
    \end{enumerate}
    It is easy for $B$ to maintain this invariant: $B$ simulates $A$
    and in any state $q\in Q_A$ and input symbol $\sigma$, if $A$
    reads a marble $c$ (as well as $\sigma$), drops a marble $c'$ (and so moves
    left) and transitions to some state $q'$, then if $B$ reads marble
    $(c,\Psi)$ in state $q$, it moves to state $q'$ and drops marble $(c',\Psi')$ where for all
    $(q_c,q_r)\in Q_{A'}^2$,
    $\delta_{D_{q_c,q_r}}(\Psi'(q_c,q_r),\sigma) = \Psi(q_c,q_r)$
    (this
    state is unique as $D_{q_c,q_r}$ is backward-deterministic). If
    $A$ lifts a pebble and moves to state $q'$, so does $B$.

Knowing the state of automata $D_{q_c,q_r}$ is important to bypass the calls to
$A'$. Indeed, suppose that $B$ reads $\vdash$ in some state $q\in
Q_A$, reading some pebble $(c,\Psi)$. Then, $B$, instead of calling
$A'$ with initial state $q_c=\deltac(q,c)$, checks whether there exists
some state $q_r$ such that $\Psi(q_c,q_r)$ is accepting for $D_{q_c,q_r}$,
in which case $B$ lifts the current pebble and moves to state
$\deltar(q,c,q_r)\in Q_A$. If no such state $q_r$ exists, then $B$
stops. By construction, we have $L(B) = L(A)$, which suffices to
conclude as $B$ is of level $1$.

Finally, we briefly justify why the non-elementary blowup is
unavoidable, by showing how to convert
    an \FO-sentence in prenex normal form, in exponential time, and more generally an \FO-formula
    $\phi(x_1,\dots,x_n)$ with $n$-free first-order variables in
    prenex normal form into an
    equivalent nested marble automaton $A_\phi$ over alphabet
    $\Sigma\times \{0,1\}^n$. The non-elementary blowup to convert
    \FO-formulas into DFA is known to be unavoidable for
    formulas in prenex normal form~\cite{DBLP:journals/apal/FrickG04}.

    Quantifier-free formulas can be converted into DFA (and a
    fortiori into nested marble automata) in exponential time. We now
    explain how to treat quantifiers.

    For a formula $\phi = \forall x\psi(x_1,\dots,x_n,x)$,
    if $A_\psi$ is a nested marble automaton of level $k$, then
    $A_\phi$ is a nested marble automaton of level $k+1$ which, for
    all input position $x$, from right to left, marks $x$ with some
    marble $1$ and the other with marble $0$, moves towards $\vdash$,
    call $A_\psi$. If $A_\psi$ returns a non-accepting state, then
    $A_\phi$ stops and rejects. If $A_\psi$ returns an accepting
    state, then $A_\phi$ moves
    back to position $x$ (which is the only one marked $1$), lifts $1$,
    and repeats the same process with position $x{-}1$, and so on until
    $x$ is the first input position. For existential quantifiers,
    while standard logic-to-automata constructions use closure under morphisms of
    automata (which yields non-determinism), they can also be dealt
    with directly using marbles. The way to proceed is the same as for
    universal quantifiers, except that the nested pebble automaton
    $A_\phi$ stops and accepts as soon as a position $x$ has been
    found. State-passing, and in particular here getting the state of
    the assistant automaton back to the top-level, is crucial. Indeed, in
    the case of existential quantifiers, the return state of the assistant
    automaton $A_\psi$ indicates to $A_\phi$ whether it must continue
    looking for some $x$ or stop. 
\end{proof}

\subsubsection{Proof of Theorem~\ref{thm:sp-removal} (formal details)}
\label{app:thm-sp-removal}

\spremoval*

\begin{proof}
There are two kinds of state-passing, through functions $\deltac$ and $\deltar$.
We deal with them separately. 

\paragraph{Removal of function $\deltac$} This function allows to pass information from level $k$ to level $k{-}1$.
More precisely, being on the left end-marker in state $q$, with marble $c$, the transducer of level $k$
aims at calling the assistant transducer in state $\deltac(q,c)$.
As the assistant transducer has access to marbles of level $k$, we can use
them to pass information. Then, the assistant transducer can go through the input
word to recover values of $q$ and $c$ to recompute $\deltac(q,c)$. 
Formally, if $T$ has set of states $Q$ and set of marbles $C$,
we will use $C\times Q$ as new set of marbles, and modify the transitions of $T$
accordingly to ensure that in every reachable configuration 
$(q,i,(c,p).v)$, we have $q=p$. Additionally, the assistant transducer $T'$
is modified as follows: it always starts in the same initial state and
goes to the left end-marker. It reads the marble $(c,q)$ of level $k$ (which is part of 
its input word), goes back to the right end-marker, and starts its computation
from state $\deltac(q,c)$. 

Regarding complexity, the modifications we described require the addition of only a constant number
of states. The alphabets of marbles are also modified, the overall complexity being polynomial.

\paragraph{Removal of function $\deltar$} This function allows to pass information from level $k{-}1$ to level $k$. 
This removal is more involved, and will induce a tower of exponential in terms
of complexity.

Intuitively, we will use Lemma~\ref{lem:regaut} to precompute the final state reached by the assistant transducer.
More precisely, we proceed by induction on $k$. If $k=0$, then the result is trivial as the two classes are 
syntactically the same in this case. Assume the result holds for $k$, and let us prove it holds for $k+1$.
Let $T = (Q,C ,..., T')$ be a $(k{+}1)$-\NMT $T$ with $n$ states, with
$T'$ a $k$-\NMT. We sketch how to build an equivalent $(k{+}1)$-\NMTSPF $\tilde{T}$. 
By induction hypothesis, there exists a $k$-\NMTSPF
$\tilde{T}'$ equivalent with $T'$. In addition, the domain of $T'$ is accepted by a $k$ nested marble automaton hence, 
thanks to Lemma~\ref{lem:regaut}, is a regular language. More precisely, 
for each final state $q'_f$, the language of words accepted by $T'$ with a run ending in 
state $q'_f$ is regular. As we did in the proof of Lemma~\ref{lem:regaut}, we consider a
backward deterministic automaton accepting it, and modify the set of marbles of $T$ to store states of these
automata. Hence, when doing its call transition in state $q$ and over marble $c$, 
$\tilde{T}$ can determine the final
state $q'_f$ reached by $T'$. It then calls $\tilde{T}'$ and directly moves to the right state $\deltar(q,c,q'_f)$.
In addition, if no such state $q'_f$ exists, then it stops. Hence, calls to the assistant transducer always terminate.

Last, we justify the fact that the non-elementary blow-up is unavoidable. It is because the domain of any state-passing free \NMT $S$ is recognizable by a finite automata of exponential size. Indeed, the calls to assistant transducers always terminate, so the domain of $S$
does not depend on assistant transducers, hence can be described by a marble automaton, hence by a finite
automaton of size exponential in the number of local states and marbles of $S$. Thus, the existence of an elementary construction for state-passing removal
would contradict the non-elementary blow-up stated in Lemma~\ref{lem:regaut}.
\end{proof}

\section{Omitted proofs of Section~\ref{sec:equivalence}}

\subsubsection{Proof of Theorem~\ref{thm:lex2marble}}

\lextomarble*

\begin{proof}
    The proof goes by induction on $k$. It is immediate for $k=0$. 
    Now, let $f:\Sigma^*\partialf \Gamma^*$ such that $f =
   \maplex_\ordera\ g$, with $\ordera=(B,\order_B)$, for some $g:(\Sigma\times B)^*\rightarrow
   \Gamma^*$ in $\LEX_k$. By induction hypothesis, there exists an
   \NMT of level $k$ $T_g$ which recognizes $g$. To obtain an \NMT $T_f$
   recognizing $f$, $T_f$ needs to enumerate all the possible
   annotations of its input according to $\order_B$, using $B$ as its
   set of marbles. Whenever a complete annotation of its input
   is obtained, it calls $T_g$ recursively. 
   We now explain how $T_f$ realizes the
   enumeration. Initially, it annotates from right to left the whole
   input with marbles $b_0$, where $b_0$ is the
   $\order_B$-minimal element. Now, suppose that $T_f$ has annotated the
   whole input (in $\Sigma$) of length $\ell$ with some $b_1\dots b_\ell\in B^*$. To obtain the successor
   annotation, it reads the marbles $b_1b_2\dots$ from left
   to right, and looks for the first position $i$ such that $b_i$ is not the
   $\order_B$-maximal element. It replaces $b_i$ by $\suc_{\ordera}(b_i)$
   and, going from right to left, replaces any $b_j, j<i$, by $b_0$,
   until it reaches the leftmarker $\vdash$. Only a constant number of
   states is needed.
\end{proof}

With a slightly more complicated proof, we can show that it is
possible to convert any lexicographic transduction directly into a
\emph{state-passing free} \NMT. Indeed, in the latter proof, no states
need to be passed to the assistant transducer $T_g$. State-passing
freeness also asks that any call to the assistant transducer $T_g$
succeeds. This is the case whenever $T_g$ is called on annotated words
$({\vdash}u{\dashv})\otimes \overline{b}$ such that $u\in\dom(f)$,
which is regular by Lemma~\ref{lem:regdomlex}. Hence, $T_f$ first
needs to check that its input is in $\dom(f)$, to ensure that it is
indeed state-passing free. 

\subsubsection{Proof of Theorem~\ref{thm:marble2lex} (formal details)}
\label{app:marble2lex}

\marbletolex*

\subsubsection{Bottom producing nested marble transducers}

As a preprocessing step, we prove that \NMTSPF can be restricted to
produce outputs only at the deepest recursive call. An $\NMT$ is
called \emph{bottom producing} if it is either of level $0$, or of
level $k>0$ with a constant output function $\out$ producing $\epsilon$ and a bottom producing assistant $\NMT$. 

\begin{lemma}\label{lem:botproducing}
    Any $\NMTSPF$ $T$ of level $k$ can be converted in \ptime into an equivalent bottom producing
    $\NMTSPF$ $T'$ of level $k$. 
\end{lemma}
\begin{proof}
Whenever $T$ produces some output word $w\in\Gamma^*$, instead $T'$ writes it on the input (as a marble), moves to ${\vdash}$, and calls a modified assistant transducer which first scans its input looking for some marble in $\Gamma^*$, which it outputs, otherwise behaves as the initial assistant transducer. The modified assistant transducer is then made bottom producing inductively. 

  Formally, the proof goes by induction on $k$. If $k=0$, it is already bottom
    producing. 
    Let $T = (\Sigma,\Gamma,C,c_0,Q,q_0,F,\delta,\out,R)$ be an
    $\NMTSPF$ of level $k$. We modify $T$ into an equivalent
    $\NMTSPF$ $T' =
    (\Sigma,\Gamma,C',c'_0,Q',q'_0,F',\delta',\out',R')$ such that
    $\out'$ is the constant function returning $\epsilon$ and $R'$
    is of level $k{-}1$. By induction hypothesis, $R'$ can be turned
    into a equivalent bottom producing $\NMTSPF$ of level $k{-}1$,
    yielding the result.

    Let us sketch the construction of $T'$. Whenever $T$ applies some transition $t\in\delta$ such that
    $w:=\out(t)\neq\epsilon$, instead, $T'$ drops a marble $w$ and
    moves to the leftmarker $\vdash$, dropping a dummy marble $\perp$ and looping in state
    $\overleftarrow{t}$ (this state indicates that $T'$ must move
    left and keeps in memory the transition $t$ which should have been
    applied). Once on the leftmarker, $T'$ calls
    the assistant transducer $R'$, which  works as follows: it scans until
    it sees, on its input, a letter of the form $(\sigma,w)\in
    \Sigma\times \Gamma^*$, in
    which case it produces $w$ and moves to an accepting state (thus
    returning to $T'$). If it sees only input symbols in $\Sigma\times
    C$, it returns to its initial position and behaves as $R$. When
    $R'$ returns, $T'$ is still in state $\overleftarrow{t}$ on
    $\vdash$. $T'$ then moves right, looping in state $\overrightarrow{t}$ until it
    sees the marble $w$, in which case it applies transition $t$, as
    $T$ initially intended. The new set of states is therefore $Q' = Q\cup
    \overleftarrow{\delta}\cup \overrightarrow{\delta}$, and the new set
    of marbles is $C' = C\cup \textsf{codom}(\out)\cup \perp$.
\end{proof}

\subsubsection{Proof of Theorem~\ref{thm:marble2lex}}

\begin{figure*}
\centering
      \begin{tikzpicture}[baseline=0, inner sep=0, outer sep=0, minimum size=0pt]
\tikzset{>={Latex[width=2mm,length=2mm]}}

\draw (-6,0) node (nodeleft) {$\vdash$}; 
\draw (6,0) node (noderight) {$\dashv$}; 
\draw (-5.8,0) -- (5.8,0);
\draw (0,0.2) node (input) {input $u$};

\draw[dotted,thick] (3.25,0) -- (3.25,-7); 
\draw (3.25,0.15) node () {$i$};


\draw (6,-1) -- (4,-1);
\draw (4,-1) arc[start angle=90, end angle=270, radius=0.25];

\draw (4,-1.5) -- (5.75,-1.5);
\draw (5.75,-1.5) arc[start angle=90, end angle=-90, radius=0.25];

\draw (5.75, -2) -- (0, -2);
\draw (0,-2) arc[start angle=90, end angle=270, radius=0.25];


\draw[->] (5,-1) -- (4.8,-1);

\draw[->] (0,-3) -- (-0.2,-3);
\draw[->] (0,-6) -- (-0.2,-6);

\draw[->] (-3,-3.5) -- (-2.8,-3.5);

\draw[->] (4.8,-1.5) -- (5,-1.5);

\draw (0,-2.5) -- (3,-2.5);
\draw (3,-2.5) arc[start angle=90, end angle=-90, radius=0.25];

\draw (3,-3) -- (-5.75,-3);
\draw (-5.75,-3) arc[start angle=90, end angle=270, radius=0.25];

\fill (-6,-3.25) circle (2pt);

\draw (-5,-2.75) node () {configuration $\nu_1$};

\draw (-5.75,-3.5) -- (-1,-3.5);
\draw (-1,-3.5) arc[start angle=90, end angle=-90, radius=0.25];

\draw (-1,-4) -- (-5.75,-4);
\draw (-5.75,-4) arc[start angle=90, end angle=270, radius=0.25];

\draw (-5.75,-4.5) -- (3,-4.5);
\draw (3,-4.5) arc[start angle=90, end angle=-90, radius=0.25];

\draw (3,-5) -- (-1.5,-5);
\draw (-1.5,-5) arc[start angle=90, end angle=270, radius=0.25];

\draw (-1.5,-5.5) -- (3,-5.5);
\draw (3,-5.5) arc[start angle=90, end angle=-90, radius=0.25];

\draw (3,-6) -- (-6,-6);

\fill (-6,-6) circle (2pt);
\draw (-5,-6.25) node () {configuration $\nu_2$};

\draw[->,dashed] (3.25,-2.75) .. controls (5, -4.25)  .. (3.25,-5.75);

\fill (3.25,-2.75) circle (2pt);
\fill (3.25,-5.75) circle (2pt);

\draw (5.54,-4.25) node () {$\begin{array}{c} \text{right-to-right} \\ \text{traversal}\end{array}$};

\draw (5,-2.75) node () {$T$ is in state $\textsf{last}_{\nu_1}(i)$};
\draw (5,-5.75) node () {$T$ is in state $\textsf{last}_{\nu_2}(i)$};

\end{tikzpicture}
\caption{Moves of $T$ illustrated, with $\nu_1\rightarrow_T^* \nu_2$\label{fig:illustrateorder}}
\end{figure*}

\begin{proof} 
 Let $f : \Sigma^*\partialf\Gamma^*$ be recognizable by some
    $\NMTSPF$ $T$ of  level $k$. We show that $f\in\LEX_k$.
    The proof goes by induction on $k$. It is immediate for $k=0$.
    We prove the induction step.

    \vspace{2mm}
    \emph{Assumptions}  Let $T = (\Sigma,\Gamma,C,c_0,Q,q_0,F,\delta,R)$ be an \NMTSPF of level $k$.
    To simplify the proof, we make a few
    assumptions on $T$, which are without loss of generality. First, we assume $T$ to be bottom
    producing thanks to Lemma~\ref{lem:botproducing} (that is why
    the output function $\out$ is not mentioned in the tuple above).


    We also assume that $T$, whenever it processes an input
    word, always start by initially doing a full right-to-left pass marking the input with
    $c_0$, and then comes back to the rightmarker
    $\dashv$. If that is not the case, then $T$ can easily be modified
    by adding a loop on state $q_0$ which always drops marble
    $c_0$ until $\vdash$ is read, then coming back to
    $\dashv$. The subtransducer $R$ is also modified in such a way
    that it does not produce anything on the first initial
    marking in $c_0^*$. This assumption will be useful
    when ordering configurations, so that initial markings will always
    be of the latter form.

   \vspace{2mm}
   \emph{Useful notions: full configurations and right-to-right traversals}
   Since $T$ is bottom-producing, only configurations whose
    reading head is placed on $\vdash$ must be considered in the
    enumeration.     A \emph{full configuration} of $T$ on some word ${\vdash}
    u{\dashv}$ is a configuration whose reading head is placed on
    $\vdash$, \ie a triple $(q,i,v)$ such that $i=0$. We denote by $\mathbb{C}_T^{f}(u)$ the set of full
    configurations on input ${\vdash}
    u{\dashv}$ \emph{reachable} from the initial configuration.

    We now define the  notion of \emph{right-to-right traversals}, which is
    standard in the theory of two-way machines. 
    Given a word
    ${\vdash}u{\dashv}$, a position $i$ and a marble $c$, we say that a pair of
    states $(q_1,q_2)$ is a right-to-right traversal (RR for short) at
    position $i$
    marked $c$, if,  informally, 
    $T$ can go from a configuration 
    $(q_1,i,cv)$ for some $v$ to the configuration $(q_2,i,cv)$ without
    visiting any position to the right of $i$. Formally, there
    must exist a sequence of configurations
    $$
\begin{array}{c}    
(q_1,i,cv)\rightarrow_T 
    (p_1,i_1,v_1cv)\rightarrow_T(p_2,i_2,v_2cv)\rightarrow_T\dots\\
    \rightarrow_T (p_m,i_m,v_mcv)\rightarrow_T (q_2,i,cv)
    \end{array}
    $$
    such that $i_j\leq i$ for all $1\leq j\leq m$. In particular, $v$
    could be substituted by any word $v'$ of same length in the
    sequence above. We denote by
    $\rr_T(u,i)\in (C\rightarrow 2^{Q^2})$ the function which
    associates with any $c\in C$, the set $\rr_T(u,i)(c)$ (just
    denoted $\rr_T(u,i,c)$), of RR
    traversals at position $i$ marked $c$. Note that for all $c,q$, it
    holds that, $(q,q)\in \rr_T(u,i,c)$.

    \vspace{2mm}
    \emph{Main argument} The main idea of the proof is to define an
    ordered alphabet $\ordera = (B,\order)$ where the set $B$ is 
    of the form $C\times B'$ for some $B'$, and 
    for all $u\in\dom(f)$ an injective mapping $\Phi_u : \mathbb{C}_T^{f}(u)\rightarrow B^*$
    such that the following properties hold:
    \begin{enumerate}
      \item for all $\nu = (q, 0, v)\in \mathbb{C}_T^{f}(u)$, $|\Phi_u(\nu)| =
        |u|+2$ and $\pi_C(\Phi_u(\nu)) =v$.  
      \item for all $\nu_1,\nu_2\in\mathbb{C}_T^f(u)$, $\nu_1\rightarrow^*_T\nu_2$ iff $\Phi_u(\nu_1)\order
        \Phi_u(\nu_2)$
      \item $\{({\vdash}u{\dashv})\otimes \Phi_u(\nu)\mid
        u\in\dom(f)\wedge \nu\in \mathbb{C}_T^{f}(u)\}$ is regular.
    \end{enumerate}

    Property $(1)$ ensures that the $\Phi_u(\nu)$ preserves the sequence of
    marbles of $\nu$. Property $(2)$ ensures that the sequence of
    successive configurations of $T$ on ${\vdash}u{\dashv}$ can be
    embedded, modulo $\Phi_u$, in the lexicographic enumeration of
    $B^{|u|+2}$. However, some words in $B^{|u|+2}$ may not correspond
    to any full configuration. Property $(3)$ ensures that this can be
    controlled with a finite automaton. If those properties are
    guaranteed, then it suffices to modify $R$ into an \NMTSPF $R'$ of
    same level which, when reading a word in $({\vdash}u{\dashv})\otimes
    \overline{b}$, where $u\in\Sigma^*$ and $\overline{b}\in B^*$, proceeds
    as follows:
    \begin{enumerate}

\item $R'$ first checks  whether $u\in\dom(f)$ and $\overline{b}=\Phi_u(\nu)$ for some
    $\nu$ a full configuration. By property $(3)$, this can be done by
    a finite automaton $A$, which is simulated as follows: $R'$ first
    moves to $\vdash$, then simulates $A$ from left to right, ending
    up in $\dashv$ in some state $q$. 

    \item if $q$ is an accepting state of $A$, then $R'$ simulates $R$ on the $C$-projection of
      $\overline{b}$, otherwise it does nothing (thus producing
      $\epsilon$). 
\end{enumerate}

By induction hypothesis, $R'$ recognizes a
    transduction $f'\in \LEX_{k{-}1}$, and as an immediate consequence of 
    $(1)$ and $(2)$, we get $f\ =\ \maplex_\ordera\ f'$, which proves
    the result.

    It remains to define $\ordera=(B,\order)$ and $\Phi_u$, and prove that they
    satisfy properties $(1)-(3)$. 
    
    \vspace{2mm}
    \paragraph{Definition of $\ordera=(B,\order)$ and $\Phi_u$} Let $u\in\dom(f)$. First, given a full
    configuration $\nu\in\mathbb{C}^f_T(u)$, a position $i\in\{0,\dots,|u|+1\}$, and a
    state $q\in Q$, we say that that $q$ is the last visited state for
    $\nu$ at position $i$, denoted $\textsf{last}(\nu,i)=q$, if,
    informally, $q$ is the last state in which $T$ was at position
    $i$, before reaching configuration $\nu$. Formally, by definition
    of $\mathbb{C}^f_T(u)$, $\nu$ is reachable from the initial
    configuration, hence there exists a finite sequence
    $\nu_0 = (q_0,|u|+1,c_0)\rightarrow_T (q_1,i_1,v_1)\rightarrow_T
    \dots\rightarrow_T (q_\ell,i_\ell=i,v_\ell)\rightarrow_T \dots
    \rightarrow_T (q_m,i_m=0,v_m)=\nu$ such that $\nu_0$ is the initial
    configuration and $\ell$ is such that $i_j<i$ for all $\ell <
    j\leq m$. We let $\textsf{last}(\nu,i) = q_\ell$.

    The set of marbles $B$ is defined as $B = C\times
    Q\times 2^{Q\times Q}$, where the second component is aimed to be
    the last states and the third component the set of RR-traversals.
    Tuples $(c_1,p_1,X_1)$ and $(c_2,p_2,X_2)$ such that $c_1\neq c_2$
    or $X_1\neq X_2$ are arbitrarily ordered (the choice of ordering
    does not matter). Otherwise, we let $(c,p_1,X)\order (c,p_2,X)$ if
    $p_1\neq p_2$ and $(p_1,p_2)\in X$.

    We now define $\Phi_u$ (for $u\in\dom(f)$). Let $\nu = (q,0,v) \in\mathbb{C}_T^f(u)$ be a full configuration
    of $T$ on ${\vdash}u{\dashv}$. By definition of
    $\mathbb{C}_T^f(u)$, $\nu$ is reachable from the initial
    configuration, hence there exists a sequence of configurations
    $\nu_0\rightarrow_T\nu_1\rightarrow\dots\rightarrow \nu$ reaching
    $\nu$ from the initial configuration $\nu_0$. We define
    $\Phi_u(\nu)$ as $v\otimes r$ where for all $0\leq i\leq |u|+1$,
    $r[i] = (\textsf{last}(\nu,i),\rr(u,i,v[i]))$.

\vspace{2mm}
    \emph{Proof of the properties of $\Phi_u$} Given $\alpha =
    (c_1,q_1,X_1)\dots (c_n,q_n,X_n)\in B^*$,
    we let $\pi_C(\alpha) =c_1\dots c_n$, $\pi_Q(\alpha)=q_1\dots q_n$
    and $\pi_\rr(\alpha) = X_1\dots X_n$. First, injectivity of
    $\Phi_u$ and property $(1)$ are immediate by definition.

    Let us show property $(2)$, which is
    the most interesting. Let $\nu_1 = (q_1,0,v_1)$ and $\nu_2 =
    (q_2,0,v_2)$.
    From right to left, suppose that
    $\Phi_u(\nu_1)\order \Phi_u(\nu_2)$. Then, either they are
    equal and so are $\nu_1$ and $\nu_2$ because $\Phi_u$ is
    injective, or they are different. In the latter case they can be
    decomposed into
    $\Phi_u(\nu_1) = w'_1(c_1,p_1,X_1)w$ and $\Phi_u(\nu_2) =
    w'_2(c_2,p_2,X_2)w$ for some $w'_1,w'_2,w\in B^*$, $p_1,p_2\in Q$,
    $c_1,c_2\in C$ such that $(c_1,p_1,X_1)\neq (c_2,p_2,X_2)$. Let $i = |w'_2| =
    |w'_1|$.

    We first show that necessarily, $c_1=c_2$. Let us assume that
    $c_1\neq c_2$. We derive a contradiction. Indeed, by
    definition of marble transducers, the only way to modify a marble
    at some position $i$ is to lift the current marble, move right and
    eventually return to position $i$ and drop another
    marble. This cannot be done while having the same suffix $w$. 
Less informally, if we denote by $\rho_1$ the configuration
    $(p_1,i,c_1\pi_C(w))$ and by $\rho_2$ the configuration
    $(p_2,i,c_2\pi_C(w))$, then either $\rho_1\rightarrow^*_T\rho_2$
    or $\rho_2\rightarrow^*_T\rho_1$. Assume that
    $\rho_1\rightarrow^*_T\rho_2$, the other case being symmetric. 
    In the sequence of
    configurations to go from $\rho_1$ to $\rho_2$, there is necessarily a configuration whose reading head
    is to the right of $i$, \ie there exists $\rho_3 = (p_3,j,v')$
    such that $j>i$ and $v'$ is a suffix of $\pi_C(w)$ such that
    $\rho_1\rightarrow_T^* \rho_3\rightarrow_T^*\rho_2$. Assume that
    $j$ is the largest integer such that such a configuration $\rho_3$
    exists. It implies that $p_3 = \textsf{last}(\nu_2,j)$. Since $v'$
    is a suffix of $\pi_C(w)$, we also get that
    $p_3=\textsf{last}(\nu_1,j)$. Therefore, there is a cycle $\rho_3\rightarrow_T^*\rho_1\rightarrow_T^*\rho_3$, contradicting
    $u\in\dom(f)$.

    Therefore, $c_1=c_2$. As the definition of $X_1$ and $X_2$ only
    depends on the input word, the position and the marble $c_1=c_2$,
    we get $X_1=X_2$. Since $(c_1,p_1,X_1)\neq (c_2,p_2,X_2)$, we then have
    $p_1\neq p_2$. Moreover,
    $c_1\pi_C(w) = c_2\pi_C(w)$. Let $v$ denote this word. Hence we
    get $v_1 = v'_1v$ and $v_2=v'_2v$ for some $v'_1,v'_2\in C^*$. 
    By definition of $\Phi_u$, $(p_1,p_2)\in \rr(u,i,c_1=c_2)$, we have
    $(p_1,i,v)\rightarrow_T^*(p_2,i,v)$. 
    Moreover, since
    $p_2 = \textsf{last}(\nu_2,i)$, we also get $(p_2,i,v)\rightarrow_T^*
    \nu_2$. Similarly, since $p_1 =
    \textsf{last}(\nu_1,i)$, we have $(p_1,i,v)\rightarrow_T^*
    \nu_1$. Therefore, since $T$ is deterministic, either $\nu_1\rightarrow_T^*\nu_2$ or
    $\nu_2\rightarrow_T^*\nu_1$ hold. Suppose that
    $\nu_2\rightarrow_T^*\nu_1$, we show a contradiction. In that
    case, we have
    $(p_1,i,v)\rightarrow_T^*\nu_2\rightarrow_T^*\nu_1$. Since
    $p_1\neq p_2$, the sequence of configurations from $\nu_2$ to
    $\nu_1$ has to visit again input position $i$ in state $p_1$, and
    hence $\nu_2\rightarrow_T (p_1,i,v)$. Therefore, there is a cycle
    from $(p_1,i,v)$ to itself, which contradicts that $u\in
    \dom(f)$. Hence $\nu_1\rightarrow_T^*\nu_2$.

    From left to right, suppose that $\nu_1\rightarrow_T^*\nu_2$ and
    $\nu_1\neq \nu_2$. Let $w_1 = \Phi_u(\nu_1)$ and $w_2 =
    \Phi_u(\nu_2)$. Let $w$ be the longest common suffix of $w_1$ and
    $w_2$. Suppose that $|w| = |w_1|=|w_2|$. Since $q_1 = \textsf{last}(\nu_1,0)$ and $q_2 =
    \textsf{last}(\nu_2,0)$, by definition of $\Phi_u$, the first
    letter of $w_1$ has $q_1$ as 2nd component, and the first letter
    of $w_2$ has $q_2$ as 2nd component, from which we get $q_1=q_2$,
    which contradicts $\nu_1\neq \nu_2$. 

    Therefore  $|w|<|w_1|=|w_2|$. Hence $w_1$ and
    $w_2$ can be decomposed into $w'_1(c_1,p_1,X_1)w$ and
    $w'_2(c_2,p_2,X_2)w$ where $(c_1,p_1,X_1)\neq (c_2,p_2,X_2)$. Using the same argument as in the right-to-left
    direction, we can show that necessarily, $c_1=c_2$ and therefore,
    $X_1=X_2$ and
    $p_1\neq p_2$. Let $c=c_1=c_2$ and $X=X_1=X_2$. 
   It remains to show that $(p_1,p_2)\in X =
    \rr(u,i,c)$. Let $x=1,2$. By definition of $p_x = \textsf{last}(\nu_x,i)$, it holds that
    $(p_x,i,c\pi_C(w))\rightarrow_T^* \nu_x$ following a sequence of
    configurations which never visit positions strictly larger than
    $i$. Therefore, since $\nu_1\rightarrow_T^*\nu_2$, either
    $\nu_1\rightarrow_T^* (p_2,i,c\pi_C(w))$ or
    $(p_2,i,c\pi_C(w))\rightarrow_T^* \nu_1$. We consider the two
    cases:
    \begin{enumerate}
      \item $\nu_1\rightarrow_T^* (p_2,i,c\pi_C(w))$: the sequence
        of configurations from $\nu_1$ to $(p_2,i,c\pi_C(w))$ cannot
        visit positions at the right of $i$, because $w$ is the
        longest common suffix of both $w_1$ and $w_2$, so moving right
        $i$ would imply modifying that suffix. Therefore, we get
        $(p_1,i,c\pi_C(w))\rightarrow_T^*\nu_1\rightarrow_T^*(p_2,i,c\pi_C(w))$,
        moreover by sequences which never visit positions at the right
        of $i$. It witnesses that $(p_1,p_2)\in \rr(u,i,c)$.

        \item $(p_2,i,c\pi_C(w))\rightarrow_T^* \nu_1$: we show that
          this case is impossible. Indeed, $p_2 =
          \textsf{last}(\nu_2,i)$, hence
          $(p_2,i,c\pi_C(w))\rightarrow_T^*\nu_2$ be a sequence
          which never visits positions $i$ again. In this sequence,
          there is necessarily $\nu_1$, because
          $\nu_1\rightarrow_T^*\nu_2$. It imples that $T$ can go from
          $\nu_1$ to $\nu_2$ be a sequence of configurations which
          never visit position $i$. This contradicts that $w$ is the
          longest suffix. 
      \end{enumerate}

      It remains to prove Property $(3)$. It is
      possible to construct a finite automaton which, given a
      word of the form $({\vdash}u{\dashv})\otimes \overline{b}$,
      checks whether $u\in\dom(f)$ and whether there exists
      $\nu\in\mathbb{C}_T^{f}(u)$ such that $\overline{b} =
      \Phi_u(\nu)$. We do not give the full details of this proof, but
      rather convey some intuitions. Checking whether $u\in\dom(f)$ is
      doable by a finite automaton as \NMT have regular domains, by
      Corollary~\ref{coro:regdomnmt}. Now, given $\overline{b} \in (C\times
      Q\times 2^{Q\times Q})^*$, the automaton has to check whether
      $\pi_C(\overline{b})$ is a reachable annotation of the input
      ${\vdash}u{\dashv}$. It is not difficult to see that it can be
      done with a nested pebble automaton which simulates $T$ until it
      sees $\pi_C(\overline{b})$ (checking equality can be done using
      a linear pass over the input). The result follows as nested
      pebble automata recognize regular languages by
      Lemma~\ref{lem:regaut}. Now, let $\overline{b} =
      (c_1,q_1,X_1)\dots (c_n,q_n,X_n)$. The finite automaton also has
      to check that $(1)$ $q_i = \textsf{last}(\nu,i)$ and $(2)$ that $X_i =
      \rr_T(u,i,c_i)$, for all $1\leq i\leq n$. This can be done by a
      nested pebble automaton which first poses a pebble on position
      $i$, and then checks the two properties above using a assistant
      nested pebble automaton. For property $(1)$, this can be done
      using a nested pebble automaton which simulates $T$. For
      property $(2)$, the right-to-right traversals can be computed
      using a nested marble automaton which, again, simulates
      $T$. In more details, this nested pebble automaton marks the
      current input position $i$ with some special marble, and picks some pair
      $(q_1,q_2)\in X_i$. It then simulates $T$ from position $i$,
      pebble $c_i$ and state $q_1$, until it reaches some $q_2$ as the
      same position. The special marble is used to check that no
      positions strictly greater than $i$ are visited. If no $q_2$ is
      seen, the simulation will either stop (in a non-accepting state)
      or loop foreover. The nested pebble automaton does this check
      for all pairs $(q_1,q_2)\in X_i$. It also needs to check that
      $X_i$ is complete, by also considering all pairs in
      $Q^2\setminus X_i$, and checking that none of them is a
      traversal. This can be done similarly. Once $X_i$ has been
      correctly verified to be equal to $\rr(u,i,c_i)$, the special
      marble is moved to the next position (to the left). 
      Overall, regularity follows from Lemma~\ref{lem:regaut}
      again. 
\end{proof}

\section{Omitted proofs of Section~\ref{sec:ppties}}

\subsubsection{Proof of Theorem~\ref{thm:precompose}}\label{app:precompose}

\precompose*

We first need the
following intermediate lemma:
\begin{lemma}\label{lem:sumclosure}
    Let $f,g:\Sigma^*\rightarrow \Gamma^*$ be two lexicographic
    transductions of level $k_f$ and $k_g$ respectively, such that $\dom(f)\cap \dom(g)=\varnothing$. Then $f+g$
    is lexicographic of level $\text{max}(k_f,k_g)$. 
\end{lemma}
\begin{proof}
    The proof is done by induction on the levels. We only show the
    induction step. Suppose that $f = \maplex_{(B_1,\order_1)}\ f'$ and
    $g=\maplex_{(B_2,\order_2)}\ g'$ for some $B_1,B_2$ disjoint and $f',g'$
    lexicographic.
    Then, let $\ordera=(B,\order)$ with $B = B_1\cup B_2$ ordered by
    $b\order c$ if $b,c\in B_1$ and $b\order_1 c$,  or $b,c\in B_2$ and $b\order_2 c$, or $b\neworder c$ where $\neworder$
    is an arbitrary linear order of $(B_1\times B_2)\cup (B_2\times
    B_1)$. Then, $f+g = \maplex_{\ordera} (f'|_{(\Sigma\times
      B_1)^*}+g'|_{(\Sigma\times B_2)^*} + C\ch \epsilon)$ where $C =
    ((\dom(f)\cup \dom(g))\otimes B^*)\setminus ((\Sigma\times B_1)^*\cup
    (\Sigma\times B_2)^*)$.  We can conclude by induction hypothesis
    and the fact that $k$-lexicographic transductions are closed under
    regular domain restriction, for any fixed $k$. 
\end{proof}

\begin{proof}[Proof of Theorem~\ref{thm:precompose}]   
    
    Any rational transduction $g$ can be decomposed as $g = h\circ e$
    where $h$ is a morphism and $e$ is a letter-to-letter rational
    transduction, in the sense that any transition outputs exactly one
    letter. This is easily seen: if $g$ is defined as a finite
    transducer, $e$ outputs the sequence of transitions, and $h$
    computes locally every output of the transitions. So, it suffices
    to show that $\LEX$ is closed under pre-composition by a
    morphism, and by a letter-to-letter rational transduction. 

    \vspace{2mm}
    \emph{Closure under morphism} Let $\psi : \Gamma\rightarrow
    \Lambda^*$ be a morphism. We show that $f\circ \psi\in \LEX$. We
    do it by induction on the level $k$ of $f$. If $f$ is simple,
    then $f\circ \psi$ is also simple. Otherwise, $f = \maplex_{\ordera}\
    f'$ for some $\ordera=(B,\order)$ and $f'$. We extend $\ordera$ into $\ordera'=(B',\order')$ with $B' = (\Lambda\times
    B)^{\leq M}$, where $M =
    \text{max}_{\gamma\in\Gamma}|\psi(\gamma)|$, and order it as
    follows: for all $v_1\otimes b_1, v_2\otimes b_2\in B'$, 
    $v_1\otimes b_1\order' v_2\otimes b_2$ if $v_1=v_2$ and $b_1\order b_2$, or $v_1\order_{\Lambda} v_2$ for $\order_{\Lambda}$ an arbitrary linear order. Then, let $L$ be
    the set of words in $(\Gamma\times
    (\Lambda\times B)^{\leq M})^*$, of the form
    $(\gamma_1,v_1\otimes b_1)\dots (\gamma_n,v_n\otimes b_n)$
    such that $v_i = \psi(\gamma_i)$ for all $i=1,\dots,n$. Define
    $f''$ the transduction which takes any word in $L$ and returns $f'((v_1\dots v_n)\otimes (b_1\dots b_n))$, otherwise it returns $\epsilon$. By definition of
    $f''$ and $\order_{B'}$, we have:
    $$
    f\circ \psi\ =\ \maplex_{\ordera'}\ f''
    $$
    We conclude by proving that $f''\in \LEX$. 
    It is a consequence of
    the fact that $(1)$ $f'' =  (f'\circ \pi)|_L + C\ch \epsilon$,
    where $C = (\Gamma\times (\Lambda\times B)^{\leq M})^*\setminus L$
    and $\pi$ is the
    morphism which maps any pair $(\gamma,v\otimes b)\in\Gamma\times
    (\Lambda\times B)^{\leq M}$ to $v\otimes b$, $(2)$
    the induction hypothesis, which implies that $(f'\circ \pi)\in \LEX$, $(3)$
    lexicographic transductions are closed under regular domain
    restriction (see Example~\ref{ex:domrestrict}) and $(4)$
    lexicographic transductions are closed under sum by Lemma~\ref{lem:sumclosure}. This does not change the level. 

    \vspace{2mm}
    \emph{Closure under letter-to-letter rational transductions} Let $g$
    be a letter-to-letter rational transduction, given by a finite
    transducer $(A,\out)$, where $A$ is a non-deterministic finite
    automaton with set of transitions $\Delta$, and
    $\out:\Delta\rightarrow \Lambda$ is a morphism. We show that
    $(f\circ g)\in \LEX$. Without loss of generality, we assume that $f=\maplex_{\ordera}\ f'$
    for some $\ordera=(B,\order),f'$ and $f' : (\Lambda\times B)^*\rightarrow \Gamma^*$ lexicographic. Then, $B$ is extended
    into $B' = B\times \Delta$, ordered as $(b_1,\delta_1)\order'
    (b_2,\delta_2)$ if $b_1\order b_2$ or $b_1=b_2$ and $\delta_1\order_\Delta
    \delta_2$ for $\order_\Delta$ an arbitrary order.

    Let $L\subseteq (\Gamma\times B')^*$ be the set of words
    $(\gamma_1,b_1,\delta_1)\dots (\gamma_n,b_n,\delta_n)$ such that
    $\delta_1\dots\delta_n$ is an accepting run of $A$ on
    $\gamma_1\dots\gamma_n$. Clearly, $L$ is regular. Then,
    $f\circ g = \maplex_{B'}\ (f'\circ \out\circ \pi_\Delta)|_L$
    where $\pi_\Delta$ is the $\Delta$-projection morphism.  The
    result follows as lexicographic transductions are closed under
    morphism and regular domain restriction. This does not change the level.
    \end{proof}

\end{document}